\documentclass[12pt]{article}

\usepackage{mkc}
\usepackage{mkc_refs}
\usepackage{stix}

\newcommand{\preg}[2][]{\mathrm{PR}_{#1}(#2)}
\newcommand{\pregs}{\mathrm{PR}}
\newcommand{\er}[2][]{\mathrm{ER}_{#1}(#2)}

\newcommand{\fer}[2][]{\mathrm{FER}_{#1}(#2)}
\newcommand{\fers}{\mathrm{FER}}
\newcommand{\ir}[2][]{\mathrm{IR}_{#1}(#2)}

\newcommand{\cir}[2][]{\mathrm{CIR}_{#1}(#2)}
\newcommand{\cirs}{\mathrm{CIR}}
\newcommand{\fcir}[2][]{\mathrm{FCIR}_{#1}(#2)}
\newcommand{\fcirs}{\mathrm{FCIR}}

\newcommand{\creg}[2][]{\mathrm{CR}_{#1}(#2)}

\newcommand\numberthis{\addtocounter{equation}{1}\tag{\theequation}}

\title{Mechanisms for a No-Regret Agent:\\Beyond the Common Prior\footnote{
		This work began as part of the 2018 Special Quarter on Data Science and Online Markets at Northwestern. We are especially grateful to Simina Br\^{a}nzei and Katya Khmelnitskaya for their early contributions to this project. We are also grateful to Eddie Dekel, Marciano Siniscalchi, and several anonymous referees for helpful comments, in addition to audiences at the 70th Midwest Theory Day and Northwestern. Jason Hartline and Aleck Johnsen were supported in part by NSF grant CCF-1618502.
	}}

\author{Modibo Camara\footnote{Department of Economics, Northwestern University. Email: modibokhane@gmail.com.}
	\quad Jason Hartline\footnote{Department of Computer Science, Northwestern University. Email: hartline@northwestern.edu.}
	\quad Aleck Johnsen\footnote{Department of Computer Science, Northwestern University. Email: aleckjohnsen@u.northwestern.edu.}}

\begin{document}

	\begin{titlepage}
	
	\maketitle
	
	\thispagestyle{empty}
	
	\begin{abstract}
		A rich class of mechanism design problems can be understood as incomplete-information games between a principal who commits to a policy and an agent who responds, with payoffs determined by an unknown state of the world. Traditionally, these models require strong and often-impractical assumptions about beliefs (a common prior over the state). In this paper, we dispense with the common prior. Instead, we consider a repeated interaction where both the principal and the agent may learn over time from the state history. We reformulate mechanism design as a reinforcement learning problem and develop mechanisms that attain natural benchmarks without any assumptions on the state-generating process. Our results make use of novel behavioral assumptions for the agent -- centered around \textit{counterfactual internal regret} -- that capture the spirit of rationality without relying on beliefs.
	\end{abstract}
	
	
\end{titlepage}
	
%
%
	
	\section{Introduction}\label{S1}

Mechanism design is a branch of economic theory concerned with the design of social institutions. It encompasses a wide range of phenomena that have historically been of interest to economists, including, but not limited to, auctions \parencite{Vickrey61, Myerson81}, matching markets \parencite{GS62, Roth82}, taxation \parencite{Mirrlees71}, contracts \parencite{SZ71, Ross73}, and persuasion \parencite{KG11}.


Despite this field's potential, it is often unclear whether and how mechanisms derived from economic theory can be implemented in practice. In particular, one modeling practice stands out as a barrier to implementation: the \textit{common prior} assumption. Many mechanism design problems are only interesting in the presence of uncertainty, and this uncertainty is typically modeled as stochasticity. The \textit{state} of the world is drawn according to some distribution and, importantly, the distribution is commonly known by the designer and all participants in the mechanism.\footnote{This assumption is limiting in two ways. First, mechanisms based on a common prior may not be practicable, because they rely on knowledge that a real-world designer is unlikely to possess. Second, even if the designer knows the distribution (resp. has beliefs), the participants may not arrive with the same knowledge (resp. share those beliefs).}

This paper will dispense with the common prior assumption. In its place, we consider a model of adversarial online learning where the principal and a single agent are learning about the state, over time, using data. The static mechanism design problem is a Stackelberg game of incomplete information. The principal chooses a policy $p$, the agent chooses a response $r$, nature chooses a state $y$, and payoffs are realized. In the online problem, this game is repeated $T$ times, where state $y_t$ is revealed at the end of period $t$. The sequence of states is arbitrary and the principal's mechanism should perform well without prior knowledge of the sequence. The principal's present choices can affect the agent's future behavior; this makes mechanism design a reinforcement learning problem in our model.

In the absence of distributional assumptions, standard restrictions on the agent's behavior, like Bayesian rationality, become toothless. In its place, we define \textit{counterfactual internal regret} (CIR) and assume that the agent obtains low CIR. This is an ex post definition of rationality that includes Bayesian rationality (with a well-calibrated prior) as a special case. We develop data-driven mechanisms that are guaranteed to perform well under our behavioral assumptions. Specifically, we prove upper bounds on the principal's regret from following our mechanism, relative to the single fixed policy that performs best in hindsight. Our results take the form of reductions from the principal's problem to robust versions of static mechanism design with a common prior.

\paragraph{Running Example.}

Bayesian persuasion is a model of strategic communication, due to \textcite{KG11}. It has received considerable attention from economists and, more recently, algorithmic game theorists (e.g. \cite{DX16}, \cite{CDHW20}). It is a useful test case for our framework because (a) it is interesting even with only one agent, (b) the optimal solution varies with the agent's beliefs, and (c) it has the potential to be widely applicable.\footnote{Bayesian persuasion has been used to study a wide range of topics, including recommendation systems \parencite{MSSW16}, traffic congestion \parencite{DKM17}, congested social services \parencite{AIM20}, financial stress-testing \parencite{GL18}, and worker motivation \parencite{ES20}.
}

Our running example is adapted from \textcite{KG11}. 
A drug company (the principal) seeks approval from a regulator (the agent) for a newly-developed drug. The state $y\in\{\mathrm{High}, \mathrm{Low}\}$ describes the drug's quality. Neither the regulator nor the company know the quality in advance. The company needs to design a clinical trial that will generate (possibly noisy) information about the drug's quality. Roughly, a trial $p$ specifies the probability $p(m,y)$ of sending a message $m$ to the regulator, conditional on the drug quality $y$. Informally, the message describes the outcome of the trial. After hearing the message, the regulator decides whether to approve the drug. The regulator receives a payoff if it approves a high-quality drug or rejects a low-quality drug. The company receives a payoff if the regulator approves, regardless of quality. Its challenge is to design a clinical trial that convinces the regulator to approve as many drugs as possible.

To predict behavior in incomplete-information games, we need to make assumptions about how the agents deal with uncertainty. The common prior is one such assumption. In our running example, the common prior would specify a probability $q\in[0,1]$ that the drug is high quality. Consider the case $q=1/3$. If the company does not run a trial -- e.g. it recommends ``approve'' in every state -- the regulator would never approve, as the drug is more likely to be low quality than high quality ex ante. If the company runs the most thorough trial possible -- e.g. it recommends ``approve' if and only if the drug is high quality -- the regulator would approve with probability $1/3$. Finally, consider the optimal trial. The optimal trial always recommends ``approve'' if the drug is high quality. If the drug is low quality, it recommends ``approve'' and ``reject'' with equal probability. After hearing ``approve'', the regulator's posterior puts equal weight on both states, and so it might as well approve. Here, the regulator approves with probability $2/3$.

\paragraph{Online Mechanism Design.}

In our model, both the company and the regulator would be learning about drug quality over time. New drugs arrive sequentially. For each drug, the company designs a clinical trial and generates a message. The regulator hears the message and decides whether to approve. Regardless of whether the drug is approved, both parties eventually learn the drug's true quality, and the next drug arrives. The company's strategy, called a \textit{mechanism}, maps the drug (i.e. state) history and the approval decision (i.e. response) history to a trial for the current drug. The regulator's strategy, called a learning algorithm or \textit{learner}, maps the drug quality history and the trial (i.e. policy) history to an approval decision for the current drug. This model is \textit{online} because the company and regulator must make decisions while the drugs are still arriving. It is \textit{adversarial} in the sense that we impose no assumptions on the sequence of drugs, and so any results (e.g. claiming that a mechanism performs well) must hold for all such sequences.


The company's problem is to develop a mechanism that performs as well as the best-in-hindsight trial. That is, the company should not regret following its mechanism relative to any simple alternative where it picks the same trial $p$ in every period. To evaluate what would have happened under an alternative sequence of trials, the company must take into account how the regulator's behavior would have changed. Therefore, the company faces a reinforcement learning problem and its benchmark corresponds to the notion of \textit{policy regret} in the literature on bandit learning with adaptive adversaries (e.g. \cite{ADT12}). 
In that setting, \textcite{ADT12} show that guaranteeing sublinear (policy) regret is generally impossible.\footnote{\textcite{ADT12} obtain positive results when the adversary satisfies a bounded memory assumption. \textcite{RH08} obtain positive results under a different kind of assumption, that the environment is sufficiently ``forgiving'' of mistakes. These papers reflect two prominent approaches in reinforcement learning: (a) restricting attention to Markov decision processes, and (b) assuming an ability to ``reset'' the problem \parencite{KMN99}.} This fact precludes a simple solution to the company's problem; we must constrain the regulator's behavior.\footnote{\textcite{ADMM18} consider policy regret in a repeated game and use the self-interest of the adaptive adversary to motivate behavioral restrictions. This is reminiscent of a literature on multi-agent reinforcement learning when the state is Markovian \parencite{BBS10, UV03, Littman94, HW98}. Unlike these papers, we do not have the ability in our model to advise all participants simultaneously.}



The standard way to constrain the regulator/agent's behavior -- i.e. to capture ``self-interest'' in the absence of a meaningful notion of ex ante optimality -- is to impose upper bounds on the agent's regret. This will be our approach as well. We build on existing no-regret assumptions, in ways that are intended to refine and better motivate those assumptions.


\paragraph{No-Regret Agents.}

Two notions of regret have been used historically: external and internal (or swap) regret (ER and IR). For example, \textcite{NST15} show how ER bounds combined with bidding data can be used to partially identify bidder valuations in a dynamic auction. \textcite{BMSW18} consider a dynamic pricing problem against no-ER agents.\footnote{In their model, the agent is ``learning'' an appropriate response to the principal's pricing strategy. If the agents use naive mean-based learners, \textcite{BMSW18} provide a mechanism that extracts the full surplus. In particular, the agent fails to anticipate the mechanism that the principal is using. As they point out, this leads to odd behavior: the agent may purchase goods at a price exceeding her valuation. In our setting, the agent does not face uncertainty with respect to the mechanism; instead, she faces uncertainty with respect to the state sequence.} Their analysis is generalized by \textcite{DSS19}, who study repeated Stackelberg games of complete information. Furthermore, the literature on no-regret learning in games has established that if agents satisfy a no-ER (resp. no-IR) property in a repeated game, the empirical distribution of their actions will converge to a coarse correlated equilibrium (resp. correlated equilibrium) \parencite{FV97, HM01, BHLR08, HST15}.

Both ER and IR can be thought of as ``non-policy'' regret, because they do not take into account how the agent's behavior affects the behavior of others. The justification for these regret bounds is that (a) they are satisfied by well-known learning algorithms (see e.g. \cite{LW94} for ER), and (b) they generalize optimality conditions associated with a stationary equilibrium. Nonetheless, these regret bounds can be problematic. Effectively, they assume that agents are (a) sophisticated enough to obtain low non-policy regret, but (b) not aware that their true objective is policy regret. Keep in mind that an agent who minimizes policy regret can easily obtain high non-policy regret, and thereby violate the regret bounds.


To avoid this problem, the principal in our model can commit to a mechanism that is \textit{nonresponsive} to the agent's behavior: the policy $p_t$ depends on the state history but not on the agent's response history.
When mechanisms are nonresponsive, non-policy regret and policy regret coincide for the agent. 
Then, bounds on the agent's regret are permissive assumptions that allow a wide range of sophisticated and self-interested behavior, including Bayesian rationality. Keep in mind, there is no need to resort to responsiveness if nonresponsive mechanisms tightly bound the principal's regret.\footnote{This approach seems spiritually similar to that of \textcite{IMSW20}, who develop mechanisms that incentivize efficient social learning. By restricting attention to simple disclosures (i.e. unbiased subhistories), they significantly simplify the agents' inferential problem and can motivate a permissive notion of frequentist rationality. Having restricted disclosure in this manner, they nonetheless design mechanisms with optimal rates of convergence.}

\paragraph{Counterfactual Internal Regret.}

Without constraints on the agent's behavior, an early mistake by the principal can result in a permanent, undesirable shift in the agent's behavior. As we will see, this can occur when the agent behaves as if she has additional information about the state of the world that is not accounted for in our description of the model. The agent can make the principal's problem infeasible if she is willing to exploit her information selectively, i.e. based on the principal's choice of policies. Unfortunately, neither no-ER nor no-IR assumptions can rule out selective use of information.



Our notion of rationality requires the agent to fully and consistently exploit her information, regardless of the principal's chosen policies. Existing benchmarks like external and internal regret cannot capture this requirement. To see why, it helps to consider the fable of the tortoise and the hare. Both animals have an hour to traverse a one-mile track. For the tortoise, this requirement is feasible and binding: finishing in time means hustling, without substantial breaks or detours. For the hare, however, the requirement is hardly restrictive: it may stop for a break, walk rather than run, or even run around in circles while still finishing the race in time. Benchmarks like external or internal regret imply reasonable behavior for an uninformed agent (i.e. the tortoise). But for an informed agent (i.e. the hare), these benchmarks are easy enough to satisfy that it may engage in all kinds of frivolous behavior -- possibly to the detriment of the principal.

The solution to our analogy is to strengthen the hare's benchmark. If the hare has to traverse the track in three minutes, it needs to hustle, like the tortoise. Similarly, if the agent has to obtain no-regret with her information as additional context, this would preclude the kind of frivolous behavior that makes the principal's problem infeasible. Of course, setting this benchmark requires us to know the nature and quality of the agent's information, just as we needed to know the top speed of the hare. The idea behind counterfactual internal regret is that we can identify the agent's information with her past behavior under counterfactual mechanisms. Intuitively, any information that is useful should eventually reveal itself through variation in behavior.

\paragraph{Main Results.}

This paper considers three variations on our model: one where the principal knows the agent's information, one where the agent has no private information, and one where the agent may have private information. In each case, we propose a mechanism and bound on the principal's regret in terms of the agent's counterfactual internal regret (CIR).

Our first mechanism is intended as a warmup. It requires oracle access to the agent's information and has poor performance in finite samples, but avoids some complications associated with information asymmetry between the principal and agent. First, the mechanism produces a calibrated forecast of the state in the current period using off-the-shelf algorithms, using the oracle as additional context for the forecast. Then, it chooses the worst-case optimal policy in a (hypothetical) $\epsilon$-robust version of the common prior game. In that game, the agent's response only needs to be $\epsilon$-approximately optimal, and the mechanism substitutes its forecast for the prior.

Theorem \ref{T1} bounds the principal's regret under this mechanism, under some restrictions on the stage game. Suppose there are $n_\mathcal{Y}$ states, $n_\mathcal{P}$ policies, and $n_\mathcal{R}$ responses. Fix a parameter $\epsilon>0$ (controlling robustness) and $\delta>0$ (controlling the fineness of a grid). Our bound is
\begin{equation}\label{E5}
\underbrace{O(\epsilon)}_\text{cost of $\epsilon$-robustness}
+\frac{1}{\epsilon}\left(
\underbrace{O(\cirs)}_\text{agent's regret}
+\underbrace{\tilde{O}\left(\frac{\delta^{1-n_\mathcal{Y}}n_\mathcal{Y}n_\mathcal{R}^{2n_\mathcal{P}}}{T^{1/4}}\right)}_\text{forecast miscalibration}
+\underbrace{O\left(\delta^{1/2}\right)}_\text{approximation error}
\right)
\end{equation}
If the agent satisfies no-CIR, i.e. $\cirs\to0$ as $T\to\infty$, then the principal's regret vanishes in $T$ as long as $\epsilon,\delta\to0$ at the appropriate rates. Moreover, the principal's average payoffs converge to a natural benchmark: what he would have obtained in a stationary equilibrium of the repeated game with a common prior (the empirical distribution conditioned on agent's information).

Our second mechanism applies when the agent is as uninformed as the principal. This mechanism is identical to the first, except its forecast does not use information revealed by the learner. We formalize ``uninformedness'' by assuming that the agent's external regret is non-negative (in conjunction with no-CIR). Theorem~\ref{T2} bounds the principal's regret under this mechanism, under some additional restrictions on the stage game. Our bound is
\begin{equation}\label{E6}
\underbrace{O(\epsilon)}_\text{cost of $\epsilon$-robustness}
+\frac{1}{\epsilon}\left(
\underbrace{O(\cirs)}_\text{agent's regret}
+\underbrace{\tilde{O}\left(\frac{\delta^{1-n_\mathcal{Y}}n_\mathcal{Y}}{T^{1/4}}\right)}_\text{forecast miscalibration}
+\underbrace{O\left(\delta^{1/2}\right)}_\text{approximation error}
\right)
\end{equation}
Compared to \eqref{E5}, this drops the exponential dependence on the number $n_{\mathcal{P}}$ of policies. This is because the principal's forecast does not need to take into account the agent's information, which significantly reduces the forecast miscalibration in finite samples.

Our third mechanism applies even when the agent is more informed than the principal. Here, we consider an ``informationally robust'' version of the stage game, due to \textcite{BM13}, where the agent receives a private signal from an unknown information structure. Like before, we formulate an $\epsilon$-robust version of this game, where the agent's response need only be $\epsilon$-approximately optimal. Our mechanism is identical to the second mechanism, except that it chooses the worst-case optimal policy in the $\epsilon$-informationally-robust game instead of the $\epsilon$-robust game.

Theorem \ref{T3} bounds the principal's regret under this mechanism, under some restrictions on the stage game. Let $\hat{\pi}_T$ denote the empirical distribution of states $y_{1:T}$. Given a common prior $\pi$, let $\nabla(\pi)$ be the difference between the principal's maxmin payoff and his maxmax payoff across all possible information structures. Roughly, our bound is
\begin{equation}\label{E7}
\underbrace{\nabla(\hat{\pi}_T)}_\text{cost of informational robustness}
+\underbrace{O(\epsilon)}_\text{cost of $\epsilon$-robustness}
+\frac{1}{\epsilon}\left(
\underbrace{O(\cirs)}_\text{agent's regret}
+\underbrace{\tilde{O}\left(\frac{\delta^{1-n_\mathcal{Y}}n_\mathcal{Y}}{T^{1/4}}\right)}_\text{forecast miscalibration}
+\underbrace{O\left(\delta^{1/2}\right)}_\text{approximation error}
\right)
\end{equation}
Unlike \eqref{E5} and \eqref{E6}, the principal's regret does not vanish as $T\to\infty$. However, it is vanishing up to the cost of informational robustness $\nabla(\hat{\pi}_T)$ that would also be present under a common prior, if the agent were more informed than the principal.

Finally, although our focus is not on computational complexity, the reader should note that the computational tractability of our mechanisms will depend critically on our ability to solve robust mechanism design problems under a common prior. So, while our bounds on the principal's regret apply to a large class of games, evaluating tractability may require a case-by-case analysis.



\paragraph{Additional Related Work.}

Within computer science, many researchers share our goal of replacing prior knowledge in mechanism design with data. The literature on sample complexity in mechanism design allows the principal to learn the state distribution from i.i.d. samples \parencite{BBHM08,CR14,MR15,Syrgkanis17}. Here, the data arrives as a batch rather than online, there is no repeated interaction and the question of responsiveness does not arise. However, there has also been work that applies online learning to auction design (e.g. \cite{DS16}; \cite{DHLSSV17}) and Stackelberg security games (e.g. \cite{BBHP15}). Here, agents are either short-lived or myopic, whereas our agent is long-lived and potentially forward-looking.

These papers can avoid the agent's learning problem because they emphasize applications where the agent does not face uncertainty, or where truthfulness is a dominant strategy. In contrast, \textcite{IMSW20, CDHW20} study problems that are closer to our own, insofar as both the principal and the agent must learn from data. They impose behavioral assumptions that are suited for i.i.d. data, whereas our model generalizes to adversarial data.

Within economics, research has focused on relaxing prior knowledge, rather than replacing it entirely. Part of the literature on robust mechanism design relaxes the common prior to some kind of approximate agreement on the distribution \parencite{MM11,OT12,AKS13,OP17,JMM12}. Our approach will suggest $\epsilon$-robustness and $\epsilon$-informational-robustness as alternatives to ``approximate agreement''.

\paragraph{Organization.}

Section \ref{S2} introduces the stage game and $\epsilon$-robustness. Section \ref{S3} introduces the repeated game. Section \ref{S4} defines external, internal, and counterfactual internal regret. Section \ref{S4} presents our mechanism and regret bounds when the agent's learner is known. As preparation for the remaining results, section \ref{S5} introduces the stage game with private signals. Section \ref{S4} presents our mechanism and regret bounds when the agent is uninformed. Section \ref{S4} presents our mechanism and regret bounds when the agent may be more informed than the principal. Section \ref{S9} concludes with a discussion of open problems.

Appendix \ref{Ap1} applies these results to two special cases: our running example, and a principal-agent problem. Appendix \ref{Ap2} considers the complexity of the agent's learning problem. Appendix \ref{Ap3} describes our forecasting algorithms in more detail. Appendix \ref{Ap4} relaxes some of the restrictions on the stage game and generalizes our results. Appendix \ref{Ap5} collects proofs.

	\section{Stage Game}\label{S2}

Our model features three participants: a male principal, a female agent, and nature. As advertised, we are interested in a repeated interaction between these participants. To begin with, however, we describe the stage game, which will constitute a single-round of the repeated game. In the stage game, the principal moves first and commits to a policy $p\in\mathcal{P}$. Next, the agent observes the policy $p$ and then chooses a response $r\in\mathcal{R}$. Utility functions depend on the response $r$, the policy $p$, and an unknown state of the world $y\in\mathcal{Y}$, chosen by nature. Formally, the agent's utility function is $U:\mathcal{R}\times\mathcal{P}\times\mathcal{Y}\to[0,1]$ while the principal's utility function is $V:\mathcal{R}\times\mathcal{P}\times\mathcal{Y}\to[0,1]$.

\begin{assume}[Regularity]\label{A1}
	We impose the following regularity conditions.
	\begin{enumerate}
		\item The state space $\mathcal{Y}$ is finite.
		\item The response space $\mathcal{R}$ is a compact space with metric $d_\mathcal{R}$.
		\item The policy space $\mathcal{P}$ is a compact space with metric $d_{\mathcal{P}}$.
		\item The utility $U$ is equi-Lipschitz continuous in $(r,p)$ for Lipschitz constants $K^U_{\mathcal{R}}$ and $K^U_{\mathcal{P}}$, i.e.
		\[
		\forall y\in\mathcal{Y}:\quad\left|U(r,p,y)-U(\tilde r,\tilde p,y)\right|\leq K^U_{\mathcal{R}}d_\mathcal{R}(r,\tilde{r})+K^U_{\mathcal{P}}d_{\mathcal{P}}(p,\tilde{p})
		\]
		\item The utility $V$ is equi-Lipschitz continuous in $(r,p)$ for Lipschitz constants $K^V_{\mathcal{R}}$ and $K^V_{\mathcal{P}}$, i.e.
		\[
		\forall y\in\mathcal{Y}:\quad\left|V(r,p,y)-V(\tilde r,\tilde p,y)\right|\leq K^V_{\mathcal{R}}d_\mathcal{R}(r,\tilde{r})+K^V_{\mathcal{P}}d_{\mathcal{P}}(p,\tilde{p})
		\]
	\end{enumerate}
\end{assume}

Later on, we will use covers to convert infinite action spaces into discrete approximations. For example, our running example involved an infinite policy space.

\begin{defn}[Covers]
	Let $\mathcal{X}$ be a metric space with metric $d_{\mathcal{X}}$. Generally, lower case letters $x$ denote elements of $\mathcal{X}$ while upper case letters $X$ denote subsets.
	\begin{enumerate}
		\item Fix $\delta_{\mathcal{X}}>0$. Let the partition $\mathcal{C}_{\mathcal{X}}$ be a $\delta_{\mathcal{X}}$ \emph{cover} of $\mathcal{X}$. That is, for every set $X\in\mathcal{C}_X$, any two elements $x,\tilde{x}\in X$ must be within distance $\delta_{\mathcal{X}}$ of one another, i.e. $d_{\mathcal{X}}(x,\tilde{x})<\delta_{\mathcal{X}}$.
		
		\item To reduce notation, we also let $\mathcal{C}_{\mathcal{X}}$ denote a discretized subset of $\mathcal{X}$. That is, for each set $X\in\mathcal{C}_{\mathcal{X}}$, choose a unique $x\in X$ to represent $X$. In that case, we say $x\in\mathcal{C}_{\mathcal{X}}$.
		
		\item Let $x\in\mathcal{X}$ and $\tilde{x}\in\mathcal{C}_{\mathcal{X}}$. We say that $\tilde{x}$ is the \emph{discretization} of $x$ if $x,\tilde{x}$ belong to the same subset $X\in\mathcal{C}_{\mathcal{X}}$.
	\end{enumerate}
\end{defn}

We will refer to covers $\mathcal{C}_{\mathcal{P}}$ of the policy space (with metric $d_{\mathcal{P}}$), $\mathcal{C}_{\mathcal{R}}$ of the response space (with metric $d_{\mathcal{R}}$), and $\mathcal{C}_{\Delta(\mathcal{Y})}$ of the state distributions $\Delta(\mathcal{Y})$ (with the $l_1$ metric).\footnote{Note that $\mathcal{P}$, $\mathcal{R}$, $\Delta(\mathcal{Y})$ are all compact. Therefore, we can always construct a finite cover.} Of course, if the underlying set $\mathcal{X}$ is finite to begin with, we can simply set $\delta_{\mathcal{X}}=0$ and let $\mathcal{C}_{\mathcal{X}}=\mathcal{X}$.

The stage game plays an important role in our analysis. Two of our results (theorems \ref{T1} and \ref{T2}) are best understood as reducing the online mechanism design problem to the simpler task of finding a ``locally-robust'' policy in the stage game. In the locally-robust problem, we maintain the traditional common prior assumption: that is, the state $y$ is drawn from a commonly known distribution $\pi$. However, we relax the assumption that the agent maximizes her expected utility $\ex[y\sim\pi]{U(r,p,y)}$. Instead, she chooses a response (or a distribution $\mu$ over responses) that guarantees her an expected utility that is within an additive constant $\epsilon$ of the optimum. Since this assumption only partially identifies the agent's behavior, the principal's utility can take on a range of values. The principal's worst-case utility from following policy $p$ is described by the function
\[
\alpha_p(\pi,\epsilon)=\min_{\mu\in\Delta(\mathcal{R})}\ex[y\sim\pi]{\ex[r\sim\mu]{V(r,p,y)}}
\quad\mathrm{s.t.}\quad
\max_{\tilde{r}\in\mathcal{R}}\ex[y\sim\pi]{U(\tilde{r},p,y)}-\ex[y\sim\pi]{\ex[r\sim\mu]{U(r,p,y)}}\leq\epsilon
\]
and his best-case utility is described by
\[
\beta_p(\pi,\epsilon)=\max_{\mu\in\Delta(\mathcal{R})}\ex[y\sim\pi]{\ex[r\sim\mu]{V(r,p,y)}}
\quad\mathrm{s.t.}\quad
\max_{\tilde{r}\in\mathcal{R}}\ex[y\sim\pi]{U(\tilde{r},p,y)}-\ex[y\sim\pi]{\ex[r\sim\mu]{U(r,p,y)}}\leq\epsilon
\]
The worst-case optimal (or $\epsilon$-robust) policy, defined below, is one of two main ingredients in our proposed mechanisms (the other is a calibrated forecasting algorithm).

\begin{defn}[$\epsilon$-Robustness]
	The \emph{$\epsilon$-robust} policy is worst-case optimal over all response distributions $\mu$ that achieve at least the agent's optimal expected utility minus $\epsilon$. Formally,
	policy is $$p^*(\pi,\epsilon)\in\arg\max_{p\in\mathcal{P}}\alpha_p(\pi,\epsilon)$$
\end{defn}

\begin{defn}[Cost of $\epsilon$-Robustness]
	Fix a distribution $\pi$ and parameter $\epsilon>0$. The cost of $\epsilon$-robustness is the distance between the principal's best-case utility (under the best-case optimal policy) and worst-case utility (under the worst-case optimal policy). Formally,
	\[
	\Delta(\pi,\epsilon)=\max_{p\in\mathcal{P}}\beta_p(\pi,\epsilon)-\alpha_{p^*(\pi,\epsilon)}(\pi,\epsilon)
	\]
\end{defn}

The cost of $\epsilon$-robustness will be a key variable in our upper bounds on the principal's regret in the repeated game. It will be convenient to assume that this cost is growing at most linearly in $\epsilon$, although this assumption is not really necessary (see appendix \ref{Ap4}).

\begin{assume}\label{A2}
	For any distribution $\pi$, $\Delta(\pi,\epsilon)=O(\epsilon)$.
\end{assume}

Finally, the following lemma will be important to our results. Suppose that the principal misjudges the agent. Instead of choosing a response that achieves at least her optimal expected utility minus $\epsilon$, the agent only achieves her optimal expected utility minus $\epsilon+\tilde{\epsilon}$, for $\tilde{\epsilon}>0$. Nonetheless, if the principal uses the $\epsilon$-robust policy, his utility degrades smoothly in the residual $\tilde{\epsilon}$.

\begin{lemma}\label{L1}
	Assume regularity (assumption \ref{A1}). For any distribution $\pi$, policy $p$, and constants $\epsilon,\tilde{\epsilon}>0$, the principal's worst-case and best-case utilities satisfy
	\[
	\alpha_p(\pi,\epsilon+\tilde\epsilon)\geq\alpha_p(\pi,\epsilon)-\frac{\tilde\epsilon}{\epsilon}
	\quad\mathrm{and}\quad
	\beta_p(\pi,\epsilon+\tilde\epsilon)\leq\beta_p(\pi,\epsilon)+\frac{\tilde\epsilon}{\epsilon}
	\]
\end{lemma}


Appendix \ref{Ap1} describes two well-known special cases of our model: Bayesian persuasion and the principal-agent problem. For each case, we provide a simple example, check that the example satisfies all relevant assumptions, and evaluate our results. In that sense, these examples serve as sanity checks for the rest of the paper, which involves assumptions and solutions that are sometimes rather abstract.

	\section{Repeated Game}\label{S3}

In the repeated game, the stage game is repeated $T$ times. In period $t$, the principal chooses policy $p_t$, the agent chooses response $r_t$, and nature chooses the state $y_t$. At the end of period $t$, the state $y_t$ is revealed to both the principal and the agent.

The agent's repeated game strategy (henceforth, \textit{learner} $L$) maps the state history $y_{1:t-1}$, the response history $r_{1:t-1}$, the policy history $p_{1:t-1}$, and the current policy $p_t$ to a distribution $\mu_t$ over responses. Formally, the response distribution in the $t$\textsuperscript{th} period is given by\footnote{The fact that the response distribution $\mu_t$ may depend on realized response history $r_{1:t-1}$ allows the learner to introduce correlation between responses across time, if desired.}
$$
L_t:\mathcal{Y}^{t-1}\times\mathcal{R}^{t-1}\times\mathcal{P}^t\to\Delta(\mathcal{R})
$$

The principal's repeated game strategy (henceforth, \textit{mechanism} $\sigma$) maps the state history $y_{1:t-1}$, the response history $r_{1:t-1}$, and the policy history $p_{1:t-1}$ to a distribution $\nu_t$ over policies. Formally, the policy distribution in the $t$\textsuperscript{th} period is given by
$$
\sigma_t:\mathcal{Y}^{t-1}\times\mathcal{R}^{t-1}\times\mathcal{P}^{t-1}\to\Delta(\mathcal{P})
$$


Our goal is to design a mechanism $\sigma^*$ that the principal would not regret using, relative to a finite set of alternative mechanisms. Regret -- which we define momentarily -- measures the gap in performance between $\sigma^*$ and the alternative mechanism $\sigma$ that performed best in hindsight, given the realized sequence of states $y_{1:T}$. We consider a simple set of alternative mechanisms, corresponding to some finite set of fixed policies $\mathcal{P}_0\subseteq\mathcal{P}$ that the principal wishes to consider.\footnote{Many of our results and definitions can be adapted to any finite set of nonresponsive mechanisms.} Formally, by a fixed policy $p$, we mean a \textit{constant mechanism} $\sigma^p$ that selects the same policy $$\sigma^p_t(y_{1:t-1},r_{1:t-1},p_{1:t-1})=p$$ in all periods $t$ and for all histories.

To define the principal's regret, we need notation for the agent's behavior under the proposed mechanism $\sigma^*$, as well as under the counterfactual mechanisms $\sigma^p$. Fix the state sequence $y_{1:T}$. Let $\mu^*_t$ describe the agent's behavior under $\sigma^*$, i.e.
\[
\mu^*_t=L_t\left(y_{1:t-1},r^*_{1:t-1},p^*_{1:t}\right)
\]
given the realized history of responses $r^*_{1:t-1}$ and policies $p^*_{1:t}$ under $\sigma^*$. Let $\mu_t^p$ describes the agent's behavior under $\sigma^p$, i.e.
\[
\mu^p_t=L_t(y_{1:t-1},r^p_{1:t-1},(\underbrace{p,\ldots,p}_\text{t times}))
\]
given the realized history of responses $r^p_{1:t-1}$ under $\sigma^p$.

\begin{defn}[Principal's Regret]\label{D3}
	The \emph{principal's regret} relative to the best-in-hindsight fixed policy $p\in\mathcal{P}_0$ is
	\[
	\preg{L,y_{1:T}}=\sup_{p\in\mathcal{P}_0}\frac{1}{T}\sum_{t=1}^T\left(\ex[r\sim\mu^p_t]{V(r,p,y_t)}-\ex[r\sim\mu^*_t]{V(r,\sigma^*(y_{1:t-1},r^*_{1:t-1},p_{1:t-1}),y_t)}\right)
	\]
\end{defn}

The mechanism $\sigma^*$ satisfies no-regret if the principal's regret is $o(1)$, i.e. it vanishes as $T\to\infty$. Recall that the no-regret mechanism design problem is infeasible without further assumptions on the learner $L$. The following proposition formalizes this simple observation.

\begin{prop}[Impossibility Result for Unrestricted Learners]\label{P1}
	In our running example, for every mechanism $\sigma^*$, there exists a learner $L$ along with a state sequence $y_{1:\infty}$ such that the principal's regret does not vanish, i.e.
	\[
	\lim_{T\to\infty}\preg{L,y_{1:T}}>0
	\]
\end{prop}

	\section{Behavioral Assumptions}\label{S4}

In this section, we develop a restriction on the learner $L$ that captures ``rational'' behavior by the agent, without requiring assumptions on the state sequence $y_{1:T}$. In particular, we build on no-regret assumptions pioneered in the literature on learning in games.

In online learning, regret measures how much better or worse off the agent would have been had she followed the best-in-hindsight ``simple'' strategy instead of her learner. Different notions of regret correspond to different definitions of simplicity. All of the regret notions used in this paper will be special cases of \textit{contextual regret}, defined as follows. Given a sequence $z_{1:T}$ of variables in some arbitrary set $\mathcal{Z}$, contextual regret considers a strategy ``simple'' if, for any two periods $t$ and $\tau$, sharing the same context $z_t=z_\tau$ implies taking the same response $r_t\neq r_\tau$. 

\begin{defn}\label{D4}
	Given a sequence $z_{1:T}$ of covariates, the agent's \emph{contextual regret} relative to a best-in-hindsight modification rule $h:\mathcal{Z}\to\mathcal{R}$ is
	\[
	\creg{p_{1:T},y_{1:T}}=\max_{h}\frac{1}{T}\sum_{t=1}^T\left(U(h(z_t),p_t,y_t)-U(r_t,p_t,y_t)\right)
	\]
\end{defn}

Note that, unlike our definition of the principal's regret, the agent's contextual regret does not take into account how changes in her past behavior would have also affected the principal's behavior. This omission is justified when the mechanism is nonresponsive.

\begin{defn}[Responsiveness]
	A mechanism $\sigma$ is \emph{nonresponsive} if $$\sigma_t(y_{1:t-1},r_{1:t-1},p_{1:t-1})=\sigma_t(y_{1:t-1},\tilde r_{1:t-1},p_{1:t-1})$$
	for any period $t$, state history $y_{1:t-1}$, policy history $p_{1:t-1}$, and response histories $r_{1:t-1},\tilde{r}_{1:t-1}$.
\end{defn}

Our mechanisms will be nonresponsive. This is a design choice, not an assumption. In restricting attention to nonresponsive mechanisms, we simplify the agent's problem and make our behavioral assumptions more credible. If our mechanisms were responsive, non-myopic agents would not necessarily satisfy no-regret as defined above. For example, an agent might decide to forgo an otherwise-optimal response if she believes said response would trigger an undesirable policy by the principal going forward.\footnote{For instance, in models of repeated sales, a buyer may refuse to purchase a good at a reasonable price if she believes that holding out will cause the seller to reduce prices in the future \parencite{DPS19, ILPT17}.} This behavior would be perfectly reasonable but could cause the agent to accumulate regret. Finally, as it turns out, even nonresponsive mechanisms can guarantee vanishing principal's regret in two of the scenarios we study (sections \ref{S5} and \ref{S7}). In these scenarios, there is limited room for responsive mechanisms to improve our guarantees.

In the remainder of this section, we define three special cases of contextual regret: \textit{external regret} (ER), \textit{internal regret} (IR), and \textit{counterfactual internal regret} (CIR).

	\subsection{External Regret}\label{S4-1}

In our model, external regret is contextual regret where the policy $p_t$ is the context in period $t$. That is, no-ER requires the agent to perform as well as the best-in-hindsight mapping from policies $p_t$ to responses $r_t$. Now, why should external regret include the policy as context? Because our stage game is an extensive form. A strategy in the stage game is not a response; it is a function from the observed policy to a response. Our definition of external regret compares the agent's performance to the best-in-hindsight strategy in the stage game.\footnote{Suppose that, instead, we compared the agent's performance to the best-in-hindsight response $r\in\mathcal{R}$. Defining external regret in this way would confound variation in policies with variation in the state, and could lead to odd behavior. For example, consider the``mean-based'' learner in \textcite{BMSW18}, which never deviates far from the response that maximizes the agent's empirical utility. In that paper, the learner engages in odd behavior, like spending more than the agent's valuation.
	
	Our definition is more similar to that of \textcite{HJNZ19}, where agents following a dashboard provided by the mechanism will best respond to an allocation rule given the empirical value distribution, rather than best respond to the empirical bid distribution. This way, the agent adapts sensibly to changes in the principal's policy.}

An immediate difficulty with defining ER is that the set $\mathcal{P}$ may be continuous.\footnote{If $\mathcal{P}$ is continuous, the policy $p_t$ may be unique in every period $t=1,\ldots,\infty$. In that case, requiring no-ER would be equivalent to requiring ex post optimality. That is unreasonably strong.} For instance, this is true in our running example. To ensure that the agent's learning problem is feasible in that case, we allow the agent to group together nearby policies according to the cover $\mathcal{C}_{\mathcal{P}}$ (defined in section \ref{S2}), and consider regret with respect to this coarser context. Of course, when the policy space $\mathcal{P}$ is finite, there is no need for this, and we can set $\mathcal{C}_{\mathcal{P}}=\mathcal{P}$.

\begin{defn}[External Regret]
	The agent's \emph{external regret (ER)} relative to the best-in-hindsight modification rule $h:\mathcal{C}_\mathcal{P}\to\mathcal{R}$ is
	\[
	\er{p_{1:T},y_{1:T}}=\max_{h}\frac{1}{T}\sum_{t=1}^T\left(U(h(p_t),p_t,y_t)-U(r_t,p_t,y_t)\right)
	\]
	Note the slight abuse of notation. By $h(p_t)$, we mean $h(P_t)$ where $P_t$ is the unique set in the partition $\mathcal{C}_{\mathcal{P}}$ that contains $p_t$.
\end{defn}


Although common in the literature (e.g. \cite{NST15}, \cite{BMSW18}), no-ER assumptions are insufficient for our problem. They do not circumvent the impossibility result (proposition \ref{P1}) that motivated us to restrict the agent's behavior in the first place. In particular, this is because they fail to rule out certain pathological behaviors. Because these pathological behaviors are clearly not in the agent's best interest, we also conclude that no-ER fails to rule out ``irrational'' behavior and is therefore not a good definition of ``rationality''. The following proposition (and its proof) clarifies the issue.

\begin{prop}[Impossibility Result for No-ER Learners]\label{P2}
	In our running example, for every mechanism $\sigma^*$, there exists a learner $L$ that guarantees no-ER on all state/policy sequences, i.e.
	\[
	\lim_{T\to\infty}\sup_{\tilde{p}_{1:T},\tilde{y}_{1:T}}\ex[L]{\er{\tilde{p}_{1:T},\tilde{y}_{1:T}}}=0
	\]
	along with a state sequence $y_{1:\infty}$ such that the principal's regret does not vanish, i.e.
	\[
	\lim_{T\to\infty}\preg{L,y_{1:T}}>0
	\]
\end{prop}

	\subsection{Internal and Counterfactual Internal Regret}\label{S4-2}

Before defining CIR, we provide a brief intuition: what went wrong with external regret? Recall the tortoise and hare analogy in the introduction. For a behavioral assumption to rule out pathological behaviors, it may have to adapt to the information of the agent (or the speed of the animal).

What do we mean by information? Implicit in most stochastic models is the idea that the state is fundamentally unpredictable. But there is no ex ante sense in which the deterministic sequence $y_{1:T}$ is predictable or not. In particular, the agent may behave as if she possesses ``private information'' about the sequence of states that goes beyond the ``public information'' inherent in the description of the model. In practice, the agent may have access to data that the principal lacks, notice a pattern that did not occur to the principal, or succeed through dumb luck. Formally, this reflects an adversary who simultaneously chooses the state sequence $y_{1:T}$ and the learner $L$ to cause the mechanism $\sigma^*$ to underperform. In particular, even though the agent may not observe $y_t$ when choosing a response $r_t$, this cannot prevent the adversary from ``correlating'' $r_t$ and $y_t$.\footnote{To be clear, this ``correlation'' is non-causal. For example, the adversary might choose a state sequence such that $y_t=1$ on even periods and $y_t=0$ on odd periods, and a learner $L$ such that $r_t=1$ on even periods and $r_t=0$ on odd periods. Empirically-speaking, there would be a correlation between the states and the responses. However, if we subsequently changed the value of state $y_t$ in some period $t$, this would not affect the response $r_t$, because the state is not observed and cannot affect the output of the learner $L$. That is, there is no causal relationship between $r_t$ and $y_t$.}

No-CIR requires the agent to consistently and fully exploit her private information. In the spirit of revealed preference, private information is identified with her behavior across counterfactual mechanisms. Intuitively, if the agent is able to distinguish between periods $t,\tau$ and finds it useful to do so, then her behavior should also differ between those two periods. If her behavior under one mechanism reveals private information, this information should also be accessible to her under a different mechanism. This logic allows us to define a purely ex post notion of rationality that does not refer to the agent's beliefs or to a distribution over state sequences.

No-CIR refines no-IR, a weaker condition that was developed in the literature on calibration (e.g. \cite{FV97}). Internal regret is contextual regret where the context is the agent's own behavior $r_{1:T}$. To ensure that the agent's learning problem is feasible when the response space $\mathcal{R}$ is infinite, we allow the agent to group together nearby responses according to the cover $\mathcal{C}_{\mathcal{R}}$, and consider regret with respect to this coarser context. Of course, when the response space $\mathcal{R}$ is finite, as in our running example, there is no need for this, and we can set $\mathcal{C}_{\mathcal{R}}=\mathcal{R}$.

\begin{defn}[Internal Regret]
	The agent's \emph{internal regret (IR)} relative to the best-in-hindsight modification rule $h:S_{\mathcal{P}}\times S_{\mathcal{R}}\to\mathcal{R}$ is
	\[
	\ir{p_{1:T},y_{1:T}}=\max_{h}\frac{1}{T}\sum_{t=1}^T\left(U(h(p_t,r_t),p_t,y_t)-U(r_t,p_t,y_t)\right)
	\]
	Like earlier, note the slight abuse of notation. By $h(p_t,r_t)$, we mean $h(P_t,R_t)$ where $(P_t,R_t)$ is the unique set in the collection $\mathcal{C}_{\mathcal{P}}\times\mathcal{C}_{\mathcal{R}}$ that contains $(p_t,r_t)$.
\end{defn}

Counterfactual internal regret is contextual regret where the context is the concatenation of: the policy $p^*_t$ under the proposed mechanism $\sigma^*$; the agent's behavior $r^*_{1:T}$ under $\sigma^*$; and her counterfactual behavior $r^p_{1:T}$ under the fixed policies $p\in\mathcal{P}_0$. The following definitions formalize this.

\begin{defn}[Information]
	Let the \emph{information partition} be
	$$
	\mathcal{I}=\underbrace{S_{\mathcal{P}}}_\text{policy $p^*_t$}
	\times\underbrace{S_{\mathcal{R}}}_\text{response $r_t^*$}
	\times\underbrace{\left(S_{\mathcal{R}}\right)^{|\mathcal{P}_0|}}_\text{responses $r^p_t$ for $p\in\mathcal{P}_0$}
	$$
	and let the \emph{information} $I_t$ in period $t$ be the unique set in $\mathcal{I}$ that satisfies
	$$
	I_t\ni\left(p^*_t,r_t^*,(r^p_t)_{p\in\mathcal{P}_0}\right)
	$$
\end{defn}

Note that, by definition, the same information $I_t$ is available to the agent regardless of whether the principal follows our mechanism $\sigma^*$ or deviates to some fixed policy $p\in\mathcal{P}_0$. Intuitively, the principal's choice of mechanism should not affect what information the agent has available.

\begin{defn}[Counterfactual Internal Regret]
	The agent's \emph{counterfactual internal regret (CIR)} relative to the best-in-hindsight modification rule $h:S_\mathcal{P}\times S_{\mathcal{R}}^{|\mathcal{P}_0|+1}\to\mathcal{R}$ is
	\[
	\cir{p_{1:T},y_{1:T}}=\max_{h}\frac{1}{T}\sum_{t=1}^T\left(U(h(I_t),p_t,y_t)-U(r_t,p_t,y_t)\right)
	\]
\end{defn}

The discussion in the proof of proposition \ref{P2} clarifies how no-CIR rules out the kinds of pathological or irrational behavior that no-ER fails to rule out. In the next section, we will see the crucial role that no-CIR plays in our proving our bounds on the principal's regret. The essential property is that, conditional on information $I_t$, the agent chooses a roughly constant response that is approximately best-in-hindsight for whichever mechanism the principal is considering.


	\section{Mechanism for an Informed Principal}\label{S5}

Our first result should be viewed as pedagogical. It bounds the principal's regret under a mechanism that requires oracle access to the agent's learner. This requirement is unrealistic and will be removed in sections \ref{S5} and \ref{S6}. Likewise, the bound itself will feature an exponential dependence on the size of the policy space. This dependence will also be removed in later sections.


\begin{defn}[Information Oracle]
	The \emph{information oracle} $\Omega_t:\mathcal{P}\to\mathcal{I}$ specifies the information $I_t$ that the learner $L$ would generate in period $t$ given any policy $p_t\in\mathcal{P}$ and the realized history.
\end{defn}

This case is a convenient starting point because it avoids the bulk of the information asymmetries between the principal and the agent that our later results need to address. That follows from the fact that any private information generated by the learner can be anticipated by the principal with access to the information oracle. This case is also a convenient point of departure from the common prior assumption because it permits a wider range of agent behavior without relaxing the principal's knowledge of said behavior. To be clear, under a common prior, the fact that the principal knows the agent's prior means that he also has precise knowledge of the agent's learner. In addition, since the agent is Bayesian, the agent does not find it beneficial to randomize and her learner will typically be deterministic. Essentially, the common prior provides an information oracle for free.

\begin{mech}\label{M1}
	Let the distribution $\pi_t$ be a forecast of the state $y_t$ generated by a calibrated forecasting algorithm that uses the agent's information as context.
	\begin{itemize}
		\item 
		Our forecasting algorithm applies a generic no-internal-regret algorithm due to \textcite{BM07} in an auxilliary learning problem where the action space consists of discretized forecasts $\pi\in\mathcal{C}_{\Delta(\mathcal{Y})}$ and the loss function is the negated quadratic scoring rule $S$. In each period, the algorithm makes a prediction $\pi_t$ and incurs loss $-S(\pi_t,y_t)$. Further details as well as rates of convergence are in appendix \ref{Ap3}.
		\item The context is the vector of outputs $\Omega(p)$ of the information oracle under discretized policies $p\in\mathcal{C}_{\mathcal{P}}$. The forecasting algorithm is run separately for each context.
	\end{itemize}
	Fix a parameter $\bar{\epsilon}>0$. In period $t$, the \emph{informed-principal mechanism} $\sigma^*$ chooses the discretization of the $\bar\epsilon$-robust policy $p^*(\pi_t,\bar{\epsilon})$ that treats the forecast $\pi_t$ as a common prior.
\end{mech}


Before stating the theorem in full, we present the reasoning behind the result and clarify the components of the regret bound, as well as the assumptions required. First, we require some additional notation. Let ``$t\in I$'' indicate that information $I$ is present in period $t$, i.e. $I_t=I$. Let $n_I=\sum_{t=1}^T\textbf{1}(t\in I)$ indicate the number of periods with information $I$. Let $\hat{\pi}_I$ be the empirical distribution conditioned on the agent having information $I$, i.e.
\[
\hat{\pi}_I(y)=\frac{1}{n_I}\sum_{t\in I}\textbf{1}(y_t=y)
\]

We begin with a straightforward but important observation: across all periods $t\in I$, the agent's response $r^*_t$ is roughly constant, as are her counterfactual responses $r^p_t$ under fixed policies $p\in\mathcal{P}_0$. By regularity \eqref{A1}, slight variations in responses have correspondingly slight impacts on the agent's and principal's utility. Suppose that these responses are exactly constant, i.e. $r_t=r_I$. Note that $p_t=p_I$ is exactly constant as well, across these time periods, for all constant mechanisms $\sigma^p$ as well as the proposed mechanism $\sigma^*$, which uses discretized policies. With everything constant, the principal's average utility across context $I$ takes on a familiar form:
\[
\frac{1}{n_I}\sum_{t\in I}V(r_I,p_I,y_t)=\ex[y\sim\hat{\pi}_I]{V(r_I,p_I,y)}
\]
Similarly, the agent's average utility is
\[
\frac{1}{n_I}\sum_{t\in I}U(r_I,p_I,y_t)=\ex[y\sim\hat{\pi}_I]{U(r_I,p_I,y)}
\]
Essentially, within each context $I$, we have recreated the stage game with common prior $\hat{\pi}_I$. The agent accumulates regret
\[
\epsilon_I=\max_{\tilde{r}}\ex[y\sim\hat{\pi}_I]{U(\tilde{r},p,y)}-\ex[y\sim\hat{\pi}_I]{U(r_I,p,y)}
\]

Under mechanism \ref{M1}, the principal chooses (roughly) the $\bar{\epsilon}$-robust policy for the forecast $\pi_t$. Suppose for the moment that the forecasts are also roughly constant for all periods $t\in I$, i.e. $\pi_t=\pi_I$. Since the forecast is calibrated and uses information $I_t$ as context, $\pi_I$ cannot be too far in the $l_1$ distance from $\hat{\pi}_I$ (this is essentially the definition of calibration, and follows from results in appendix \ref{Ap3}). It follows from regularity that the $\bar{\epsilon}$-robust policy for $\pi_I$ is nearly $\bar{\epsilon}$-robust for $\hat{\pi}_I$.

At this point, the principal has (roughly) applied the $\bar\epsilon$-robust policy for the empirical distribution $\hat{\pi}_I$, to an agent that obtains regret $\epsilon_I$. In that sense, the principal has misjudged the agent's capacity to make mistakes. However, recall lemma \ref{L1}: this affects the principal's best-case and worst-case utilities by at most $\epsilon_I/\bar{\epsilon}$. It follows that, roughly-speaking, the principal's utility is not much worse than the worst-case optimal utility. At the same time, it cannot be much better than the best-case optimal utility. More precisely,
\begin{equation}\label{E1}
\max_{\tilde{p}}\beta_{\tilde{p}}(\hat{\pi}_I,\bar\epsilon)+\frac{\epsilon_I}{\bar{\epsilon}}
\geq\ex[y\sim\hat{\pi}_I]{V(r_I,p_I,y)}
\geq\max_{\tilde{p}}\alpha_{\tilde{p}}(\hat{\pi}_I,\bar\epsilon)-\frac{\epsilon_I}{\bar{\epsilon}}
\end{equation}
By assumption \ref{A2}, the difference between the upper bound and the lower bound is
\begin{equation}\label{E2}
O(\bar\epsilon)+O\left(\frac{\epsilon_I}{\bar{\epsilon}}\right)
\end{equation}
This pins down the principal's utility under mechanism \ref{M1}. Moreover, the upper bound in \eqref{E1} also applies to any constant mechanism $\sigma^p$ for $p\in\mathcal{P}_0$. Therefore, \eqref{E2} also bounds the regret accumulated by the principal in context $I$.

This brings us to our key assumption: the agent's CIR is at most some constant $\epsilon$.

\begin{assume}[Bounded CIR]\label{A3}
	Let $y_{1:T}$ be the realized state sequence and let $p^*_{1:T}$ be the policy sequence generated by the proposed mechanism $\sigma^*$. There exists a constant $\epsilon\geq 0$ such that
	\[
	\epsilon\geq\cir{y_{1:T},p^*_{1:T}}\quad\mathrm{and}\quad \epsilon\geq\cir{y_{1:T},\underbrace{p,\ldots,p}_\text{$t$ times}},\:\forall p\in\mathcal{P}_0
	\]
\end{assume}

\begin{remark}
	It is worth emphasizing that this bound applies only to the realized state sequence $y_{1:T}$. That is, the agent does not need to perform well over all state sequences, and her objective need not be worst-case regret minimization. If the agent is Bayesian, for example, she will obtain low CIR as long as her beliefs are well-calibrated.
\end{remark}

Since CIR is contextual regret with information $I_t$ as context, bounded CIR ensures that
\[
\epsilon\geq\frac{1}{T}\sum_{I\in\mathcal{I}}n_I\epsilon_I
\]
Combine this with our bound \eqref{E2} on the agent's regret $\epsilon_I$ in the context of information $I$, and it follows that the principal's regret is bounded above by
\[
O(\bar\epsilon)+O\left(\frac{\epsilon}{\bar{\epsilon}}\right)
\]

To transform this intuition into a result, we need to address an assumption made along the way: that the forecast $\pi_t$ is roughly constant across all periods $t\in I$. This is not necessarily true. The adversary can choose a sequence of states $y_{1:T}$ that makes the principal appear more informed than the agent. Indeed, variation in forecasts can be interpreted as private information of the principal, even if it is spurious. On the other hand, any variation in $\pi_t$ that affects the policy $p_t$ will also be included in the agent's information $I_t$. What remains is variation in $\pi_t$ that does not affect the policy -- useless information from the principal's perspective, but not necessarily useless to the agent. If the principal expects the agent to exploit this information and the agent does not, this can lead to a suboptimal policy choice.

The following assumption restricts attention to stage games where this problem does not arise; that is, the agent's failure to exploit information that is useless to the principal does not affect the principal's utility. In appendix \ref{Ap4}, we avoid this restriction by instead assuming that the principal -- using our publicly-announced mechanism -- is not more informed than the agent.



\begin{assume}\label{A4}
	Let $\epsilon>0$. Let $\pi$ and $\tilde{\pi}$ be distributions in the stage game. If the $\epsilon$-robust policies under $\pi$ and under $\tilde{\pi}$ are close to one another, then they are also close to the $\epsilon$-robust policy under any convex combination of these distributions. Formally, for any $\lambda\in[0,1]$,
	\[
	d_{\mathcal{P}}\left(p^*(\pi,\epsilon),p^*\left(\lambda\pi+(1-\lambda)\tilde\pi,\epsilon\right)\right)=O\left(d_{\mathcal{P}}\left(p^*(\pi,\epsilon),p^*(\tilde\pi,\epsilon)\right)\right)
	\]
\end{assume}

The following theorem formalizes the preceding discussion and bounds the principal's regret under mechanism \ref{M1}.

\begin{theorem}\label{T1}
	Assume regularity (assumption \ref{A1}), restrictions on the stage game (assumptions \ref{A2}, \ref{A4}), and $\epsilon$-bounded CIR (assumption \ref{A3}). Let $\sigma^*$ be the mechanism \ref{M1}. Given access to the information oracle, for any constant $\bar{\epsilon}>0$, the principal's expected regret $\ex[\sigma^*]{\preg{L,y_{1:T}}}$ is at most
	\begin{align*}
	\underbrace{O(\bar{\epsilon})}_\text{cost of $\bar{\epsilon}$-robustness}
	+\frac{1}{\bar{\epsilon}}\cdot\left(
	\underbrace{O(\epsilon)}_\text{agent's regret}
	+\underbrace{\tilde{O}\left(T^{-1/4}\sqrt{|\mathcal{Y}||\mathcal{C}_{\Delta(\mathcal{Y})}||\mathcal{C}_R|^{(|\mathcal{P}_0|+|\mathcal{C}_{\mathcal{P}}|)/2}}\right)}_\text{forecast miscalibration}
	+\underbrace{O\left(\delta_{\Delta(\mathcal{Y})}^{1/2}\right)+O(\delta_{\mathcal{R}})+O(\delta_{\mathcal{P}})}_\text{discretization error}
	\right)
	\end{align*}
\end{theorem}

\begin{remark}Here are a few comments on this result.
	\begin{enumerate}
		\item The bound depends on the size of the partitions $\mathcal{C}_{\mathcal{P}}$, $\mathcal{C}_{\mathcal{R}}$, and $\mathcal{C}_{\Delta\mathcal{Y}}$. However, if we define these partitions to be as small as possible, we can replace these terms with the covering numbers of $\mathcal{P}$, $\mathcal{R}$, and $\Delta(\mathcal{Y})$, respectively. In that sense, our finite sample bounds will deteriorate as one increases the complexity of the action and state spaces.
		
		\item Furthermore, if we define these partitions to be the smallest possible, then theorem \ref{T1} implies that the principal's regret vanishes if $T\to\infty$ and $\epsilon,\bar{\epsilon},\delta_{\Delta(\mathcal{Y})},\delta_{\mathcal{P}},\delta_{\mathcal{R}}\to 0$ at the appropriate rates. It also follows from the proof that the principal's payoffs converge to a natural benchmark: what he would have obtained in a stationary equilibrium of the repeated game where it is common knowledge that $y_t$ is drawn independently from the empirical distribution $\hat{\pi}_{I_t}$. Formally,
		\[
		\frac{1}{T}\sum_{t=1}^TV(r_t,p_t,y_t)
		-\frac{1}{T}\sum_{I\in\mathcal{I}}n_I\max_{p\in\mathcal{P}}\beta_p(\hat{\pi}_I,0)
		\to0
		\]
		
		\item Finally, note the exponential dependence on the number of alternative mechanisms $|\mathcal{P}_0|$ and the size of the policy space cover $|\mathcal{C}_{\mathcal{P}}|$. This dependence, which is not present in theorems \ref{T2} and \ref{T3}, reflects the fact that the mechanism \ref{M1} uses the agent's information $I_t$ as context for its forecast $\pi_t$. Since our bound is uniform across all learners that satisfy $\epsilon$-bounded CIR on the realized state sequence $y_{1:T}$, it must accommodate learners that generate a lot of information, regardless of whether that information is useful. As mentioned at the beginning of this section, this is another reason why the ``informed principal'' setting seems less compelling than the settings studied in sections \ref{S7} and \ref{S8}.
	\end{enumerate}
\end{remark}

	\section{Stage Game with Private Signals}\label{S6}

In general, we cannot expect the principal to have access to an information oracle. Fortunately, we can still construct mechanisms $\sigma^*$ that obtain vanishing or bounded principal's regret without any knowledge of the learner. However, in order to state the relevant assumptions (sections \ref{S7} and \ref{S8}) and describe the mechanism (section \ref{S8}), we need to consider scenarios where the agent has private information that the principal lacks. This requires a brief detour. In this section, we revisit the stage game in order to introduce terminology that reflects agent's private information.

Suppose that the state $y$ is drawn from a known distribution $\pi$, but the agent has access to a private signal $I\in\mathcal{I}$ generated by the \textit{information structure} $\gamma$.

\begin{defn}[Information Structure]
	An \emph{information structure} is a function $\gamma:\mathcal{I}\times\mathcal{Y}\to[0,1]$ where $\gamma(\cdot,y)$ is a probability distribution over $\mathcal{I}$.
\end{defn}

The game proceeds as follows. First, nature chooses a hidden state $y\sim\pi$. Second, the principal chooses a policy $p$. Third, the agent observes a signal $I\sim\gamma(\cdot,y)$ and chooses a response $r_I$. For instance, if the agent maximizes her expected utility, her responses after signals $I$ would be
\[
r_I\in\arg\max_{\tilde{r}_I\in\mathcal{R}}\ex[y\sim\pi]{\ex[I\sim\gamma(\cdot,y)]{U(\tilde{r}_I,p,y)}}
\]
Finally, the state $y$ is revealed and payoffs are determined. 

As in section \ref{S2}, suppose the agent does not necessarily maximize her expected utility. Instead, she chooses responses $r_I$ (or distributions $\mu_I$ over responses) that guarantees her an expected utility that is within an additive constant $\epsilon$ of the optimum. For a given information structure $\gamma$, the principal's worst-case utility from following policy $p$ is described by
\begin{align*}
	\alpha_p(\pi,\gamma,\epsilon)&=\min_{\mu_I\in\Delta(\mathcal{R})}\ex[y\sim\pi]{\ex[I\sim\gamma(\cdot,y)]{\ex[r\sim\mu_I]{V(r,p,y)}}}\\
	\mathrm{subject\:to}\quad
	\max_{\tilde{r}_I\in\mathcal{R}}&\ex[y\sim\pi]{\ex[I\sim\gamma(\cdot,y)]{U(\tilde{r}_I,p,y)}}-\ex[y\sim\pi]{\ex[I\sim\gamma(\cdot,y)]{\ex[r\sim\mu_I]{U(r,p,y)}}}\leq\epsilon
\end{align*}
and his best-case utility is described by
\begin{align*}
	\beta_p(\pi,\gamma,\epsilon)&=\max_{\mu_I\in\Delta(\mathcal{R})}\ex[y\sim\pi]{\ex[I\sim\gamma(\cdot,y)]{\ex[r\sim\mu_I]{V(r,p,y)}}}\\
	\mathrm{subject\:to}\quad
	\max_{\tilde{r}_I\in\mathcal{R}}&\ex[y\sim\pi]{\ex[I\sim\gamma(\cdot,y)]{U(\tilde{r}_I,p,y)}}-\ex[y\sim\pi]{\ex[I\sim\gamma(\cdot,y)]{\ex[r\sim\mu_I]{U(r,p,y)}}}\leq\epsilon
\end{align*}
Note that $\alpha(\pi,\epsilon)$, the worst-case utility in the stage game without a private signal, is equivalent to $\alpha(\pi,\gamma,\epsilon)$ when $\gamma$ is uninformative. The same applies to $\beta$.

Recall that our theorem \ref{T1} could be interpreted as reducing the online mechanism design problem to the simpler task of finding a $\epsilon$-robust policy in the stage game without a private signal. The same is true of our next result, theorem \ref{T2}. In contrast, theorem \ref{T3} reduces the online problem to solving for a robust policy when the agent has a private signal generated by an unknown information structure. This corresponds to notion of informational robustness introduced by \textcite{BM13} and applied by \textcite{BBM17}, applied to our single-agent setting.

\begin{defn}[$\epsilon$-Informational-Robustness]
	The worst-case optimal (or \emph{$\epsilon$-informationally-robust}) policy for an unknown information structure $\gamma$ is
	\[
	p^\dagger(\pi,\epsilon)\in\arg\max_{p\in\mathcal{P}}\inf_{\gamma}\alpha_p(\pi,\gamma,\epsilon)
	\]
\end{defn}

\begin{defn}[Cost of $\epsilon$-Informational-Robustness]
	Fix a distribution $\pi$ and parameter $\epsilon>0$. The cost of \emph{$\epsilon$-informational-robustness} is the distance between the principal's best-case utility (under the best-case optimal policy (for the best-case information structure) and worst-case utility (under the worst-case optimal policy for the worst-case information structure). Formally,\footnote{Why do we evaluate the cost of informational robustness under the worst-case information structure? Because the regret guarantee that we obtain in theorem \ref{T3} applies uniformly across all learners $L$. As we will see, different learners will induce different empirical information structures $\gamma$. Our cost of informational robustness must accommodate the worst-case information structure, which loosely corresponds to the worst-case learner.}
	\[
	\nabla(\pi,\epsilon)
	=\max_{p\in\mathcal{P}}\sup_{\gamma}\beta_p(\pi,\gamma,\epsilon)
	-\max_{p\in\mathcal{P}}\inf_{\gamma}\alpha_p(\pi,\gamma,\epsilon)
	\]
\end{defn}

Let $\nabla(\pi)=\nabla(\pi,0)$ denote the cost of informational robustness in the traditional setting where the agent is optimizing exactly ($\epsilon=0$). It will be convenient to assume that the cost is growing at most linearly in $\epsilon$, although this assumption is not really necessary (see appendix \ref{Ap4}).

\begin{assume}\label{A5}
	For any distribution $\pi$, $\nabla(\pi,\epsilon)=\nabla(\pi)+O(\epsilon)$.
\end{assume}

Finally, we verify that lemma \ref{L1} still applies in the presence of private signals.

\begin{lemma}\label{L2}
	Assume regularity (assumption \ref{A1}). For any distribution $\pi$, information structure $\gamma$, policy $p$, and constants $\epsilon,\tilde{\epsilon}>0$, the principal's worst-case and best-case utilities satisfy
	\[
	\alpha_p(\pi,\gamma,\epsilon+\tilde\epsilon)\geq\alpha_p(\pi,\gamma,\epsilon)-\frac{\tilde\epsilon}{\epsilon}
	\quad\mathrm{and}\quad
	\beta_p(\pi,\gamma,\epsilon+\tilde\epsilon)\leq\beta_p(\pi,\gamma,\epsilon)+\frac{\tilde\epsilon}{\epsilon}
	\]
\end{lemma}

	\section{Mechanism for an Uninformed Agent}\label{S7}

Our second result bounds the principal's regret under a mechanism that does not require detailed knowledge of the learner $L$. Instead, this result assumes that the agent is not more informed than the principal. To begin, the mechanism is as follows.

\begin{mech}\label{M2}
	Let the distribution $\pi_t$ be a forecast of the state $y_t$.
	\begin{itemize}
		\item Our forecasting algorithm applies a generic no-internal-regret algorithm due to \textcite{BM07} in an auxilliary learning problem where the action space consists of discretized forecasts $\pi\in\mathcal{C}_{\Delta(\mathcal{Y})}$ and the loss function is the negated quadratic scoring rule.
	\end{itemize}
	Fix a parameter $\bar{\epsilon}>0$. In period $t$, the \emph{uninformed-agent mechanism} $\sigma^*$ chooses the discretization of the $\bar\epsilon$-robust policy $p^*(\pi_t,\bar{\epsilon})$ that treats the forecast $\pi_t$ as a common prior.
\end{mech}

What does it mean for an agent to be uninformed? Following the intuition developed in section \ref{S4}, the agent's behavior cannot reveal an understanding of the state sequence that goes far beyond the principal's forecast. This can be formalized by adding a lower bound on the agent's ER to our upper bound on the agent's (counterfactual) IR.\footnote{Although they study a different problem, \textcite{BGLS18} also use lower bounds on ER to prove results, exploiting the fact that exponential weights guarantees non-negative expected ER \parencite{GM16}.}

 \begin{assume}[Lower-Bounded ER]\label{A6}
	Let $y_{1:T}$ be the realized state sequence and let $p^*_{1:T}$ be the policy sequence generated by the proposed mechanism $\sigma^*$. There exists a constant $\tilde\epsilon\geq 0$ such that
	\[
	\er{y_{1:T},p^*_{1:T}}\geq-\tilde\epsilon\quad\mathrm{and}\quad\er{y_{1:T},\underbrace{p,\ldots,p}_\text{$t$ times}}\geq-\tilde\epsilon,\:\forall p\in\mathcal{P}_0
	\]
\end{assume}

While there is no a priori sense in which the deterministic sequence $y_{1:T}$ is predictable or not, this combination of bounds can be seen as an ex post definition of unpredictability. Intuitively, if an agent fully exploits the information she reveals under the proposed mechanism $\sigma^*$ (no-IR) without outperforming the best use of public information (non-negative ER), her private information cannot be particularly useful. Fully exploiting useless information generally means ignoring it.

To see this, suppose the policy $p$ is fixed and that the learner obtains non-positive IR and non-negative ER.  It is trivial to show that IR is non-negative and bounded below by ER, so it follows that the learner's IR and ER both equal zero. In turn, IR and ER can only be equal when the best-in-hindsight responses conditional on the context (i.e. the learner's response) are the same in every context. That is, the context is useless. To achieve zero IR, the learner's response must equal some best-in-hindsight response conditional on the context. If the best-in-hindsight response is unique, this means that the learner's response is the same in every period.

What this amounts to, essentially, is that our reasoning for theorem \ref{T1} largely applies to theorem~\ref{T2}. Let us recall the first steps of that argument. Previously, we considered all periods $t\in I$ with information $I$ as context. It followed immediately from the definition of information that the agent's responses $r_t$ were roughly some constant $r_I$. Furthermore, since the principal's forecasts used $I_t$ as context, the constant policy $p_I$ was calibrated to the empirical distribution $\hat{\pi}_I$.

Now, our mechanism does not have access to $I_t$ and is not calibrated to $\hat{\pi}_I$. Instead, for every policy context $P\in\Sigma_{\mathcal{P}}$, it is calibrated to the empirical distribution $\hat{\pi}_P$ conditioned on $p_t\in P$. Formally,
\[
\hat{\pi}_P(y)=\frac{1}{n_P}\sum_{t\in P}\textbf{1}(y_t=y)
\]
where $t\in P$ indicates $p_t\in P$ and $n_P$ is the number of periods $t\in P$. The policy context $P$ is coarser than information $I$, by definition of the latter. So, the principal behaves as if the agent shares his prior $\hat{\pi}_P$, while the agent behaves as if she receives $I$ as a private signal.

This is where non-negative ER comes in. The agent's information $I$ is useless to her. If there is a unique best-in-hindsight response within policy context $P$, then the agent will choose roughly the same response $r_t=r_P$ in every period $t\in P$. In other words, the policy context $P$ coincides with the agent's information $I$, and the principal is correct in assuming that the agent (roughly) optimizes against the empirical distribution $\hat{\pi}_P$. Our previous argument goes through.

Again, we just assumed that there is a unique best-in-hindsight response within policy context $P$. What if this is not the case, i.e. the best-in-hindsight response is not unique? In general, the argument breaks down. The agent can condition her action on her private information $I$, which no longer necessarily coincides with $P$. To be clear, this private signal $I$ remains useless to the agent. Moreover, the $\bar{\epsilon}$-robust policy is by definition robust to multiplicity of best responses. However, if the agent's best response is correlated with the state, this can undermine the principal's utility even if it does not affect the agent's.\footnote{For example, consider a stage game with a binary response $r\in\{0,1\}$, a binary state $y\in\{0,1\}$, and a binary policy $p\in\{\mathrm{Risky},\mathrm{Safe}\}$. The agent's utility is always zero. The principal's utility under the risky policy is $1$ if $r=y$ and $-1$ otherwise. It is slightly negative under the safe policy. If $y$ is drawn from the uniform distribution, and the agent optimizes without a signal, then the principal prefers the risky policy. If the agent receives a signal that is perfectly correlated with the state, and sets $r=1-y$, then the principal prefers the safe policy.}

The following assumption restricts attention to stage games where this issue does not arise. Informally, it asserts that if a private signal is useless to the agent, then it has limited relevance to the principal, assuming that the principal is following (nearly) optimal policies.  Formally, the value of information structure $\gamma$ to the agent in the stage game with common prior $\pi$ and policy $p$ is
\[
\phi_p(\pi,\gamma)=\max_{r,r_I\in\mathcal{R}}\ex[y\sim\pi]{\ex[I\sim\gamma(\cdot,y)]{U(r_I,p,y)}-U(r,p,y)}
\]
This is the expected utility of the agent that optimizes given information structure $\gamma$ minus the expected utility of the agent if she does not receive a private signal.

\begin{assume}\label{A7}
	Let $\pi$ be a distribution, $\epsilon>0$ be a constant, and $\gamma$ be an information structure (intuitively, one that is not useful to the agent).
	\begin{enumerate}
		\item If the principal uses $\epsilon$-robust policy $p^*(\pi,\epsilon)$, his maxmin payoff without $\gamma$, i.e. $\alpha_{p^*(\pi,\epsilon)}(\pi,\epsilon)$, is not much larger than his maxmin payoff with $\gamma$, i.e. $\alpha_{p^*(\pi,\epsilon)}(\pi,\gamma,\epsilon)$. That is,		
		\[
		\alpha_{p^*(\pi,\epsilon)}(\pi,\epsilon)-\alpha_{p^*(\pi,\epsilon)}(\pi,\gamma,\epsilon)
		=
		O\left(\phi_{p^*(\pi,\epsilon)}(\pi,\gamma)\right)+O(\epsilon)
		\]
		\item The principal's maxmax payoff with $\gamma$ under any policy $p\in\mathcal{P}$, i.e. $\beta_p(\pi,\gamma,\epsilon)$, is not much larger than his maxmax payoff without $\gamma$ under the best-case policy, i.e. $\max_{\tilde{p}\in\mathcal{P}}\beta_{\tilde{p}}(\pi,\epsilon)$. That is,
		\[
		\beta_p(\pi,\gamma,\epsilon)-\max_{\tilde{p}\in\mathcal{P}}\beta_{\tilde{p}}(\pi,\epsilon)
		=
		O\left(\phi_p(\pi,\gamma)\right)+O(\epsilon)
		\]
	\end{enumerate}
\end{assume}
Both parts of assumption \ref{A7} would be immediate if the information structure $\gamma$ were uninformative, because the left-hand sides would be non-positive. Basically, we require useless (to the agent) private signals to be similar to uninformative private signals in these two respects.

Finally, we are ready to bound the principal's regret under mechanism \ref{M2}.

\begin{theorem}\label{T2}
	Assume regularity (assumption \ref{A1}), restrictions on the stage game (assumptions \ref{A2}, \ref{A4}, \ref{A7}), $\epsilon$-bounded CIR (assumption \ref{A3}), and $\tilde{\epsilon}$-lower-bounded ER (assumption \ref{A6}). Let $\sigma^*$ be the uninformed-agent mechanism \ref{M2}. For any constant $\bar{\epsilon}>0$, the principal's expected regret $\ex[\sigma^*]{\preg{L,y_{1:T}}}$ is at most
	\begin{align*}
		\underbrace{O(\bar{\epsilon})}_\text{cost of $\bar{\epsilon}$-robustness}
		+\underbrace{O(\tilde{\epsilon})}_\text{agent's information}
		+\frac{1}{\bar{\epsilon}}\cdot\left(
		\underbrace{O(\epsilon)}_\text{agent's regret}
		+\underbrace{\tilde{O}\left(T^{-1/4}\sqrt{|\mathcal{Y}||\mathcal{C}_{\Delta(\mathcal{Y})}|}\right)}_\text{forecast miscalibration}
		+\underbrace{O\left(\delta_{\Delta(\mathcal{Y})}^{1/2}\right)+O(\delta_{\mathcal{R}})+O(\delta_{\mathcal{P}})}_\text{discretization error}
		\right)
	\end{align*}
\end{theorem}

\begin{remark}
	If we define the partition $\mathcal{C}_{\Delta\mathcal{Y}}$ to be the smallest possible, then theorem \ref{T2} implies that the principal's regret vanishes if $T\to\infty$ and $\epsilon,\bar{\epsilon},\tilde{\epsilon},\delta_{\Delta(\mathcal{Y})},\delta_{\mathcal{P}},\delta_{\mathcal{R}}\to 0$ at the appropriate rates. It also follows from the proof that the principal's payoffs converge to a natural benchmark: what he would have obtained in a stationary equilibrium of the repeated game where it is common knowledge that $y_t$ is drawn independently from the empirical distribution $\hat{\pi}_{P_t}$. Formally,
	\[
	\frac{1}{T}\sum_{t=1}^TV(r_t,p_t,y_t)
	-\frac{1}{T}\sum_{P\in\mathcal{C}_{\mathcal{P}}}n_P\max_{p\in\mathcal{P}}\beta_p(\hat{\pi}_P,0)
	\to0
	\]
\end{remark}

	\section{Mechanism for an Informed Agent}\label{S8}

In section \ref{S4}, we assumed that the principal knows the agent's learner $L$. The implication of this assumption is that the principal is as informed as the agent. In section \ref{S5}, we assumed that the agent is as uninformed as the principal. In this section, we allow the agent to be more informed than the principal. This generality comes at a cost: we no longer ensure vanishing principal's regret. Instead, we show that, in the limit, the following mechanism guarantees regret that is no greater than the cost of informational robustness.

\begin{mech}\label{M3}
	Let the distribution $\pi_t$ be a forecast of the state $y_t$.
	\begin{itemize}
		\item Our forecasting algorithm applies a generic no-internal-regret algorithm due to \textcite{BM07} in an auxilliary learning problem where the action space consists of the discretized forecasts $\pi\in\mathcal{C}_{\Delta(\mathcal{Y})}$ and the loss function is the negated quadratic scoring rule.
	\end{itemize}
	Fix a parameter $\bar{\epsilon}>0$. In period $t$, the \emph{informed-agent mechanism} $\sigma^*$ chooses the discretization of the $\bar\epsilon$-informationally-robust policy $p^\dagger(\pi_t,\bar{\epsilon})$ that treats the forecast $\pi_t$ as a common prior.
\end{mech}

Theorem \ref{T3} builds on the same reasoning as theorems \ref{T1} and \ref{T2}. First, we need to adapt assumption \ref{A4} to the case with private signals.

\begin{assume}\label{A8}
	Let $\epsilon>0$. Let $\pi$ and $\tilde{\pi}$ be distributions in the stage game. If the $\epsilon$-informationally-robust policies under $\pi$ and under $\tilde{\pi}$ are close to one another, then they are also close to the $\epsilon$-informationally-robust policy under any convex combination of these distributions. Formally, for any $\lambda\in[0,1]$,
	\[
	d_{\mathcal{P}}\left(p^\dagger(\pi,\epsilon),p^\dagger\left(\lambda\pi+(1-\lambda)\tilde\pi,\epsilon\right)\right)=O\left(d_{\mathcal{P}}\left(p^\dagger(\pi,\epsilon),p^\dagger(\tilde\pi,\epsilon)\right)\right)
	\]
\end{assume}

Next, recall how, in the previous section, we were concerned that the principal's policy $p_t$ in period $t$ was calibrated to the empirical distribution $\hat{\pi}_P$ given policy context $P\in\mathcal{C}_{\mathcal{P}}$ (where $t\in P$) rather than the empirical distribution $\hat{\pi}_I$ given information $I=I_t$. There, we resolved that problem by assuming the agent was uninformed (non-negative ER). Here, our solution is even simpler: choose a policy $p_t$ that is robust to the agent's private information $I$, whatever that may be.

To be more precise, recall that the policy context $P$ is coarser than information $I$. We can interpret periods $t\in I$ as those periods in which the agent received a private signal $I$. By looking at the frequency of information $I$ within policy context $P$, we can define an empirical information structure $\hat{\gamma}_P$ using Bayes' rule, i.e.
\[
\hat{\gamma}_P(I,y)=\frac{n_I\hat{\pi}_I(y)}{n_P\hat{\pi}_P(y)}\cdot\textbf{1}(I\subseteq P)
\]
where $I\subseteq P$ is shorthand for $t\in I\implies t\in P$. Before, we could roughly approximate principal's and agent's utility as their expected utility in the stage game where the state $y$ was drawn from the empirical distribution $\hat{\pi}_I$. Now, the approximation is the expected utility in the stage game where $y\sim\hat{\pi}_P$ and the agent receives private signal $I$ from the empirical information structure $\hat{\gamma}_P$. Of course, the principal's policy $p_t$ is robust to all information structures $\gamma$, including $\hat{\gamma}_P$.

Next, we formalize this discussion and bound the principal's regret under mechanism \ref{M3}.

\begin{theorem}\label{T3}
	Assume regularity (assumption \ref{A1}), restrictions on the stage game (assumptions \ref{A5}, \ref{A8}), and $\epsilon$-bounded CIR (assumption \ref{A3}). Let $\sigma^*$ be the informed-agent mechanism \ref{M3}. For any constant $\bar{\epsilon}>0$, the principal's expected regret $\ex[\sigma^*]{\preg{L,y_{1:T}}}$ is at most
	\begin{align*}
		\underbrace{\frac{1}{T}\sum_{P\in \mathcal{C}_\mathcal{P}}n_P\nabla(\hat{\pi}_P)
		+O(\bar{\epsilon})}_\text{cost of $\bar{\epsilon}$-informational-robustness}
		+\frac{1}{\bar{\epsilon}}\cdot\left(
		\underbrace{O(\epsilon)}_\text{agent's regret}
		+\underbrace{\tilde{O}\left(T^{-1/4}\sqrt{|\mathcal{Y}||\mathcal{C}_{\Delta(\mathcal{Y})}|}\right)}_\text{forecast miscalibration}
		+\underbrace{O\left(\delta_{\Delta(\mathcal{Y})}^{1/2}\right)+O(\delta_{\mathcal{R}})+O(\delta_{\mathcal{P}})}_\text{discretization error}
		\right)
	\end{align*}
\end{theorem}

\begin{remark}
	In contrast to our previous results, this regret bound does not vanish. However, if we define the partition $\mathcal{C}_{\Delta\mathcal{Y}}$ to be the smallest possible, the bound does converge to
	\[
	\frac{1}{T}\sum_{P\in\mathcal{C}_\mathcal{P}}n_P\nabla(\hat{\pi}_P)
	\]
	as $T\to\infty$ and $\epsilon,\bar{\epsilon},\delta_{\Delta(\mathcal{Y})},\delta_{\mathcal{P}},\delta_{\mathcal{R}}\to 0$ at the appropriate rates. This is the best possible guarantee in a stationary equilibrium of the repeated game where (a) it is common knowledge that $y_t$ is drawn independently from the empirical distribution $\hat{\pi}_{P_t}$ and (b) the agent has access to an unknown information structure $\gamma$.
\end{remark}

	\section{Conclusion}\label{S9}

We studied single-agent mechanism design where the common prior assumption is replaced with repeated interaction and frequent feedback about the world.  Our primary motivation was to remove a barrier (the common prior) that makes it difficult to implement mechanisms in practice. However, we also want to emphasize that this work can be viewed as a learning foundation for (robust) mechanism design. Indeed, our results show that policies similar to those predicted by a common prior can perform well even without making any assumptions about the data-generating process. This lends credibility to researchers who invoke the common prior for tractability, but do not expect it to be taken literally. However, there are two caveats.
\begin{enumerate}
	\item Our policies are robust to agents that behave suboptimally by up to some $\epsilon>0$. In contrast, most papers on local robustness involve an optimizing agent with misspecified beliefs (e.g. \cite{MM11,OT12,AKS13}). These notions coincide sometimes but not always. In addition, our policies sometimes require informational robustness \parencite{BM13}.
	
	\item The number of interactions $T$ required for our mechanisms to approximate the static common prior game may be large. In particular, our bounds depend on features of the stage game, like the size of the policy and response spaces, and the number of states. These features may also affect the agent's learning rate, which in turn affects our bounds. In that sense, the common prior assumption may be less appealing in more games that are more complex.
\end{enumerate}

\paragraph{Further Work.}

There are several directions in which to generalize and improve this work. To begin, it is not clear whether our finite sample bounds have a tight dependence on the number of periods $T$ and various other parameters. For example, is it possible to remove the exponential dependence in theorem \ref{T1} on the size of the policy space?\footnote{One approach the principal might take is to attempt to discern the agent's beliefs from the description of her learner $L$, and substitute those beliefs for his own forecast. If successful, this would tie the principal's forecast miscalibration to the agent's counterfactual internal regret.} In addition, there may be opportunities for tightening our results in less abstract settings where the stage game has more structure.

Our analysis was restricted to single-agent problems. Suppose there are multiple agents. From the perspective of any one agent, her opponents correspond to adaptive adversaries (c.f. \cite{ADMM18}) whose future behavior is influenced by the agent's present response. However, if the number of participants is large and the mechanism's outcome preserves the differential privacy of each agent's response history (c.f. \cite{MT07}), the behavioral assumptions developed here may also be suitable for the multi-agent setting.

We assumed that the principal and agent observe the state after every interaction, but this may be unrealistic in many applications. For instance, in contract theory the state is a function from the agent's actions to outcomes. Let us briefly refer to the principal-agent problem in appendix \ref{Ap1-2}. There, if the agent chooses to work, we do observe whether the project succeeds or not. However, we may not learn whether the project would have succeeded had the agent shirked. To mitigate this issue, we could consider the case with bandit feedback, where participants observe their own payoffs but not the state. The challenge with bandit feeback is that it requires responsive mechanisms, as the mechanism must depend on the principal's payoffs, which in turn depend on the agent's response.\footnote{Relatedly, \textcite{BBHP15} consider a repeated Stackelberg game where the state is the agent's private type. The principal receives bandit feedback: he never observes the type directly but can infer it from the agent's behavior. The issues associated with responsiveness do not arise in this model as the agent is myopic (or more precisely, there is a sequence of short-lived agents).}

In section \ref{S8}, where the agent may be more informed than the principal, the principal's regret did not vanish but rather converged to the cost of informational robustness under a common prior. There is reason to believe that this result is not tight. Although the principal will never have access to the private signal $I$ of the agent, he may attempt to learn (via the agent's past behavior) about the information structure $\gamma$ that generates it. In turn, the agent may anticipate this and attempt to manipulate the principal's policy by feigning (partial) ignorance of her private signal. This suggests a less conservative definition of informational robustness, where the principal learns the quality of any information that the agent decides to exploit. However, in the repeated game, this approach would require responsive mechanisms.

As the last two paragraphs illustrate, we need a theory of behavior for responsive mechanisms. The no-regret conditions used here and elsewhere are not as well-motivated when the mechanism (or adversary) is responsive, insofar as they do not generalize traditional rationality assumptions. Extending the logic of no-regret conditions to a larger set of mechanisms -- but not necessarily all mechanisms -- is a clear priority for further work.

	\printbibliography
	
	\newpage\appendix
	
	\section{Special Cases}\label{Ap1}
	
	\subsection{Bayesian Persuasion}\label{Ap1-1}

There is an informed sender (i.e. principal) and an uninformed receiver (i.e. agent). The principal designs the process by which information is revealed to the agent. Let $\mathcal{M}$ be a finite set of messages that he can send. Knowing that the agent will react to an informative message, the principal attempts to persuade the agent towards actions that he prefers. Let $\mathcal{A}$ be a finite set of actions that the agent can take. The agent chooses a response $r:\mathcal{M}\to\mathcal{A}$ that maps messages to actions.

A policy is an information structure $p:\mathcal{Y}\to\Delta(\mathcal{M})$. That is, an information structure $p_t(y_t)$ describes the probability of a message $m_t$ being sent, conditional on the state being $y_t$. The agent receives the message $m_t$ and takes action $a_t=r_t(m_t)$. While the agent may not know the process that generated the state $y_t$, she understands the process $p_t$ that generates the message $m_t$ conditional on the state. Armed with this understanding, she can infer something about the state $y_t$ based on the message $m_t$.

All that remains is to specify payoffs. Let $u:\mathcal{A}\times\mathcal{Y}\to\reals$ be the agent's utility function from a given action in a given state. Similarly, let $v:\mathcal{A}\times\mathcal{Y}\to\reals$ be the principal's utility. In the previous subsection, the utility functions $U,V$ depended on the triple $(r,p,y)$ rather than the pair $(a,y)$. To reconcile the two models, we let participants evaluate $(r,p,y)$ by their expected utility conditional on the state. Formally,
\[
U(r,p,y)=\sum_{m\in\mathcal{M}}p(m,y)\cdot u(r(m),y)\quad\mathrm{and}\quad V(r,p,y)=\sum_{m\in\mathcal{M}}p(m,y)\cdot v(r(m),y)
\]
When the state is fixed, the residual variation in utility is due to the fact that messages are drawn randomly from the distribution $p(y)$. These distributions are common knowledge because the agent observes the principal's policy $p$ before taking an action. Indeed, the fact that the principal commits to an information structure is the defining feature of Bayesian persuasion.

\begin{example}[Judge-Prosecutor Game]
	The state space is $\mathcal{Y}=\{\mathrm{Innocent},\mathrm{Guilty}\}$ and the action space is $\mathcal{A}=\{\mathrm{Convict},\mathrm{Acquit}\}$. The judge has 0-1 utility $u$ and prefers to convict if the defendant is guilty and acquit if the defendant is innocent. Regardless of the state, the prosecutor's utility $v$ is 1 following a conviction and 0 following an acquittal.
	
\end{example}


This example satisfies regularity \eqref{A1} with the discrete metric on $\mathcal{R}$, the $l_1$-metric on $\mathcal{P}$, and $K^U_\mathcal{R}=K^V_\mathcal{R}=K^U_\mathcal{P}=K^V_\mathcal{P}=1$.

The worst-case policy $p^*(\pi,\epsilon)$ sends the message ``convict'' whenever the defendant is guilty. If the defendant is innocent, it sends the message ``convict'' with probability
\[
q=\max\left\{1,\min\left\{0,\frac{p-\epsilon}{1-p}\right\}\right\}
\]
The cost of $\epsilon$-robustness $\Delta(\pi,\epsilon)=O(\epsilon)$ decreases smoothly with $\epsilon$. This game satisfies assumption \ref{???} since $q$ is increasing in $p$ (and hence convex combinations of distributions $p$ will yield $q$ that is bounded between the $\epsilon$-robust policies for the extremal distributions, which are close by assumption).

The worst-case policy $p^\dagger(\pi,\epsilon)$ for an unknown private signal is full transparency. The cost of informational robustness, i.e. $\nabla(\pi,0)$, is the difference between the principal's value under the common prior $\pi$ and his payoff under full transparency. This game satisfies assumption \ref{???} with $M_1=1$ and $M_2=O(\epsilon)$. It trivially satisfies assumptions \ref{???} and \ref{???} since $p^\dagger$ is constant.

	\subsection{Contract Design}\label{Ap1-2}

In classic models of moral hazard, the principal incentivizes an agent to put effort into a task the principal cares about. The timing of the game is as follows: (1) the principal commits to a contract, (2) the agent takes a hidden action, (3) nature randomly chooses an outcome, (4) the agent is paid based on the outcome, (5) the game concludes. For concreteness, we consider the limited liability model due to \textcite{Sappington83} where both participants are risk-neutral but the principal is not allowed to charge the agent. This model has been popularized by recent work in robust contract design (see e.g. \cite{Carroll15}, \cite{DRT19}).

Formally, let $\mathcal{R}$ be a finite set of actions that the agent can take. Let $\mathcal{O}$ be a finite set of outcomes $o$. The principal observes the outcome but not the action. The state $y:\mathcal{R}\to\mathcal{O}$ describes how actions map to outcomes. The employer commits to a contract $p:\mathcal{O}\to[0,\bar{p}]$ that specifies a non-negative payment for each outcome. The cost function $c:\mathcal{R}\to\reals$ describes how costly it is for the agent to take a particular action. The agent's utility function is
\[
U(r,p,y)=p(y(r))-c(r)
\]
The benefit function $b:\mathcal{O}\to\reals$ describes how beneficial a given outcome is to the principal. The principal's utility function is
\[
V(r,p,y)=b(y(r))-p(y(r))
\]
Through the contract $p$, the principal can incentivize the agent to take actions that, depending on the state, will lead to a more beneficial outcome.

\begin{example}\label{Ex5}
	The agent is given a task of unknown difficulty. There are two actions $\mathcal{A}=\left\{\mathrm{work},\mathrm{shirk}\right\}$, two outcomes $\mathcal{O}=\{\mathrm{success}, \mathrm{failure}\}$, and three states $\mathcal{Y}=\{\mathrm{trivial}, \mathrm{moderate}, \mathrm{impossible}\}$. In the trivial state, both actions lead to success. In the impossible state, both actions lead to failure. In the moderate state, work leads to success and shirk leads to failure.
	
	The principal's benefits are $b(\mathrm{success})=2$ and $b(\mathrm{failure})=0$. The agent's costs are $c(\mathrm{work})=1$ and $c(\mathrm{shirk})=0$. In the impossible and trivial states, the optimal contract pays nothing after both outcpmes and the agent will shirk. In the moderate state, the optimal contract pays $p(\mathrm{success})=5$ to cover the agent's costs if she works, otherwise $p(\mathrm{failure})=0$. Generally, if the principal pays the agent after success, the agent will have to take into account the risk that the task turns out to be impossible (where work induces costs without any payment) or trivial (where work is not required for payment). To incentivize work, the contract must compensate the agent accordingly.
\end{example}


This example satisfies regularity \eqref{A1} with $U,V$ normalized, the discrete metric on $\mathcal{R}$, the sup-norm-metric on $\mathcal{P}$, and $K^U_\mathcal{R}=K^V_\mathcal{R}=K^U_\mathcal{P}=K^V_\mathcal{P}=1$.

The worst-case policy $p^*(\pi,\epsilon)$ sets $p(\mathrm{failure})=0$ and
\[
p(\mathrm{success})=\frac{c(\mathrm{work})-c(\mathrm{shirk})+\epsilon}{\pi(\mathrm{moderate})}
\]
so long as $p(\mathrm{success})\leq\bar{p}$ and the principal's $\pi$-expected payoff is greater than zero when the agent works. Otherwise, the worst-case policy sets all transfers to zero. The cost of $\epsilon$-robustness $\Delta(\pi,\epsilon)=O(\epsilon)$ decreases smoothly with $\epsilon$.

This game satisfies assumption \ref{???}. To see this, note that as long as working is strictly more costly than shirking, the optimal policies that induce effort are bounded away from the optimal policies that do not. Among the policies that do not induce effort, convex combinations of the distribution will not make inducing effort desirable. Among policies that do induce effort, the fact that the payments following success are decreasing in $\pi(\mathrm{moderate})$ means (as in the last section) that convex combinations of distributions lead to optimal policies that are between the extremal policies.

The worst-case policy $p^\dagger(\pi,\epsilon)$ for an unknown private signal is the same as the optimal policy under a common prior without a private signal. The cost of informational robustness, i.e. $\nabla(\pi,0)$, is the difference between the principal's value when the agent only works in the ``moderate'' state and the principal pays her cost of effort conditional on success (assuming the principal prefers this to shirking with zero transfers) and his value in the common prior game without a private signal. This game satisfies assumption \ref{???} with $M_1=O(\epsilon)$ and $M_2=0$.

	\section{Agent's Learning Problem}\label{Ap2}

%
%
%
%
%
%
%
%
%

Upper bounds on external regret are often viewed as compelling assumptions (e.g. \cite{NST15}, \cite{BMSW18}) because there exist relatively simple algorithms that guarantee vanishing ER as $T\to\infty$. For example, the exponential weights algorithm (a.k.a. hedge algorithm, exponentiated gradient algorithm) satisfies no-ER. In contrast, our behavioral assumptions -- e.g. no-FCIR -- may appear daunting, insofar as the agent must solve a learning problem with a context space that is exponential in the number of alternative mechanisms, $|\Sigma_0|$. When both the sequence of states $y_{1:T}$ and the learner $L$ are particularly pathological, no-FCIR may indeed be too strong an assumption. When the the learner satisfies additional properties, or sequence of states has some stochastic structure (e.g. is i.i.d. or Markov), no-FCIR may be more reasonable.

In this section, we make one simple observation. There exists a learner that guarantees no-FCIR (and hence no-CIR) for the agent under our mechanism from theorem \ref{T2}, with the best rate of convergence we can hope for. 

\begin{defn}[CFL]
	Suppose the principal publicizes the forecast $\pi_t$ in every period $t$.\footnote{In our view, part of the principal's objective is to make the agent's problem as simple as possible. From a worst-case perspective, there is no benefit to hiding this information. With that said, we see no reason why this result should not apply under the weaker assumption that the principal's forecasting algorithm is public knowledge.} The \emph{common forecast learner} (CFL) sets
	\[
	r_t\in\arg\max_{r\in\mathcal{R}}\ex[y\sim\pi_t]{U(r_t,p_t,y_t)}
	\]
\end{defn}

Proposition \ref{P3} verifies that the CFL satisfies the behavioral assumptions of theorem \ref{T2}.

\begin{prop}\label{P3}
	Let $\sigma^*$ be the mechanism from theorem \ref{T2}. Then the CFL satisfies $\epsilon$-bounded FCIR \eqref{A4} in expectation, i.e.
	\[
	\ex[L,\sigma^*]{\fcirs}\leq\epsilon=\tilde{O}\left(\frac{1}{T^{1/4}}\sqrt{|\mathcal{Y}||\mathcal{F}|}\right)+O\left(\sqrt{|\mathcal{Y}|\delta_{\mathcal{F}}}\right)
	\]
	and $\tilde{\epsilon}$-lower-bounded FER \eqref{A5}, where $\tilde{\epsilon}=0$. Moreover, if the agent uses CFL, the principal's regret bound in theorem \ref{T2} applies regardless of whether alignment \eqref{A10} holds.
\end{prop}

Note that these rates preserve the $T^{1/4}$ convergence rate (up to $\delta_{\mathcal{F}}$ error) that is present in all of our mechanisms and reflects miscalibration of the principal. In that sense, the fact that the agent is also learning does not deteriorate the principal's performance at all. Although this has not been our emphasis so far, it would be interesting to see whether (or identify conditions under which) other simple learning algorithms satisfy our behavioral assumptions with decent rates of convergence.

	\section{Calibrated Forecasting}\label{Ap3}

In this appendix, we describe our forecasting algorithm and bound its miscalibration.

A linearly homogeneous, differentiable function $H$ is \textit{strongly convex} with parameter $\xi$ if
\[
H(\pi)\geq H(\tilde{\pi})+\nabla H(\tilde{\pi})\cdot(\pi-\tilde{\pi})+\frac{\xi}{2}\Vert\pi-\tilde{\pi}\Vert^2_2
\]
The gradient of $H$ describes a \textit{proper scoring rule} $S(\pi)=\nabla H(\pi)$ where $H(\pi)=\pi\cdot S(\pi)$ \parencite{McCarthy56}. A scoring rule $S:\Delta(\mathcal{Y})\to\reals^{\mathcal{Y}}$ is proper if the report $\tilde{\pi}$ that maximizes the $\pi$-expected score is the distribution $\pi$. Strong convexity of $H$ can be thought of as sharpening the incentives for truth-telling \parencite{Boutilier12}.

Specifically, consider the quadratic scoring rule (see e.g. \cite{JNW08})
\[
S_y(\pi)=2\pi(y)-\sum_{\tilde{y}\in\mathcal{Y}}\pi(\tilde{y})^2
\]
where $H(\pi)=\Vert\pi\Vert^2_2$ is strongly convex with $\xi=2$.

Recall that the mechanism $\sigma^*$ is supposed to be nonresponsive. As a consequence, we cannot determine the principal's beliefs $\pi_t$ in a given period based on his historical payoffs. To ensure that the beliefs $\pi_t$ are well-calibrated, we consider an auxilliary online learning problem based on a scoring rule $S$. In period $t$, the principal makes a prediction $\pi_t$ with loss function $S_{y_t}(\pi_t)$. Specifically, the predictions come from the discretized set of priors $\mathcal{F}_1$, formed by choosing a representative element $\pi$ from each set in the partition $\mathcal{S}_{\mathcal{F}}$. In terms of the score, this approximation has limited cost. Let $\pi=\left[\hat{\pi}_{F}\right]_{\mathcal{F}_1}$ be the belief $\pi\in\mathcal{F}_1$ that is closest to the empirical distribution $\hat{\pi}_{F}$. Then
\begin{align*}
	S_y(\hat{\pi}_{F})-S_y(\pi)
	&\leq S_y(\hat{\pi}_{F})-S_y(\hat{\pi}_{F}-\delta_{\mathcal{F}})\\
	&=2\hat{\pi}_{F}(y)-\sum_{\tilde{y}\in\mathcal{Y}}\hat{\pi}_{F}(\tilde{y})^2-2(\hat{\pi}_{F}(y)-\delta_{\mathcal{F}})+\sum_{\tilde{y}\in\mathcal{Y}}(\hat{\pi}_{F}(\tilde{y})-\delta_{\mathcal{F}})^2\\
	&=2\delta_{\mathcal{F}}-\sum_{\tilde{y}\in\mathcal{Y}}\hat{\pi}_{F}(\tilde{y})^2+\sum_{\tilde{y}\in\mathcal{Y}}(\hat{\pi}_{F}(\tilde{y})-\delta_{\mathcal{F}})^2\\
	&\leq2\delta_{\mathcal{F}}\numberthis\label{E39}
\end{align*}
where $\hat{\pi}_{F}-\delta_{\mathcal{F}}$ is shorthand notation for the vector $\left(\hat{\pi}_{F}(y)-\delta_{\mathcal{F}}\right)_{y\in\mathcal{Y}}$.

In this auxilliary problem, the exponential weights algorithm (see e.g. \cite{CL06}) obtains expected external regret at most
\[
\sqrt{2T\log |S_\mathcal{F}|}
\]
relative to the best-in-hindsight $\pi^*_{F}\in\mathcal{F}_1$. A reduction due to \textcite{BM07} (theorem 5) translates this into a bound on expected internal regret of
\[
|S_\mathcal{F}|\sqrt{2T\log |S_\mathcal{F}|}
\]
relative to the best-in-hindsight $\pi^*_{F}\in\mathcal{F}_1$. Combine this with the maximum approximation error \eqref{E39} to bound the expected internal regret relative to the best-in-hindsight contextual belief $\pi\in\Delta(\mathcal{Y})$, which must be the empirical distribution $\hat{\pi}_{F}$ since $S$ is proper. Specifically,
\begin{equation}\label{Ekappa}
|S_\mathcal{F}|\sqrt{2T\log |S_\mathcal{F}|}+2T\delta_{\mathcal{F}}\geq\ex{\sum_{\pi\in\mathcal{F}_1}n_{F}\hat{\pi}_{F}\cdot\left(S(\hat{\pi}_{F})-S(\pi)\right)}
\end{equation}
where $n_{F}$ is the number of periods $t$ where $\left[\pi_t\right]_{\mathcal{F}_1}=F_t$. This is a statement about the expected scoring loss, where the expectation reflects randomization in the algorithm. Our next result, lemma \ref{L12}, translates this into a statement about the $l_1$ distance between the principal's belief $\pi$ and the empirical distribution $\hat{\pi}_{F}$.

\begin{lemma}\label{L12}
	Let $S$ be a proper scoring rule where the optimal expected score $H$ is $\xi$-strongly convex. Then
	\[
	\sqrt{\frac{2|\mathcal{Y}|\kappa}{\xi}}\geq\frac{1}{T}\sum_{\pi\in\mathcal{F}_1}n_{F}d_1(\pi,\hat{\pi}_{F})
	\]
	where
	\[
	\kappa=\frac{1}{T}\sum_{\pi\in\mathcal{F}_1}n_{F}\hat{\pi}_{F}\cdot\left(S(\hat{\pi}_{F})-S(\pi)\right)
	\]
\end{lemma}
\begin{proof}
	Consider the principal's $\pi$-expected regret from predicting $\tilde{\pi}$:
	\begin{align*}
		\pi\cdot\left(S(\pi)-S(\tilde\pi)\right)
		&=H(\pi)-\pi\cdot\nabla H(\tilde{\pi})\\
		&\geq H(\tilde{\pi})-\nabla H(\tilde{\pi})\cdot\tilde{\pi}+\frac{\xi}{2}\Vert\pi-\tilde{\pi}\Vert^2_2\\
		&=\frac{\xi}{2}\Vert\pi-\tilde{\pi}\Vert^2_2\\
		&\geq\frac{\xi}{2}\left(\frac{1}{\sqrt{|\mathcal{Y}|}}\Vert\pi-\tilde{\pi}\Vert_1\right)^2\\
		&=\frac{\xi}{2|\mathcal{Y}|}\Vert\pi-\tilde{\pi}\Vert_1^2
	\end{align*}
	where the second-to-last line follows from $\Vert\cdot\Vert_1\leq|\mathcal{Y}|^{1/2}\Vert\cdot\Vert_2$. It follows that his regret in the auxilliary problem satisfies
	\[
	\kappa\geq\frac{1}{T}\sum_{\pi\in\mathcal{F}_1}n_{F}\frac{\xi}{2|\mathcal{Y}|}d_1\left(\pi,\hat{\pi}_{F}\right)^2
	\]
	where, implicitly, $F$ is the forecast context such that $\pi\in F$. Take the square root of both sides of this inequality:
	\begin{align*}
		\sqrt{\kappa}
		&\geq\sqrt{\frac{1}{T}\sum_{\pi\in\mathcal{F}_1}n_{F}\frac{\xi}{2|\mathcal{Y}|}d_1\left(\pi,\hat{\pi}_{F}\right)^2}\\
		&\geq\frac{1}{\sqrt{T}}\cdot\frac{1}{\sqrt{T}}\sum_{\pi\in\mathcal{F}_1}n_{F}\sqrt{\frac{\xi}{2|\mathcal{Y}|}}d_1\left(\pi,\hat{\pi}_{F}\right)\\
		&\geq\sqrt{\frac{\xi}{2|\mathcal{Y}|}}\cdot\frac{1}{T}\sum_{\pi\in\mathcal{F}_1}n_{F}d_1\left(\pi,\hat{\pi}_{F}\right)
	\end{align*}
	where the first line is the $l^2$ norm of a vector with $T$ entries, the second line is the $l^1$ norm of that same vector, and the inequality follows from $\Vert\cdot\Vert_1\leq T^{1/2}\Vert\cdot\Vert_2$. Collapse these inequalities and rearrange terms to obtain the desired result.
\end{proof}

If we use the quadratic scoring rule, this lemma implies
\begin{equation}\label{E4}
\ex{\frac{1}{T}\sum_{\pi\in\mathcal{F}_1}n_{F}d_1(\pi,\hat{\pi}_{F})}
\leq
\sqrt{|\mathcal{Y}||S_\mathcal{F}|\sqrt{\frac{2\log |S_\mathcal{F}|}{T}}+2|\mathcal{Y}|\delta_{\mathcal{F}}}
\end{equation}
To optimize this bound, up to log factors, set $\delta_{\mathcal{F}}=\left(\frac{1}{T}\right)^{\frac{1}{2|\mathcal{Y}|+2}}$, assuming $|S_\mathcal{F}|=\left(\left(\frac{1}{\delta_{\mathcal{F}}}\right)^{|\mathcal{Y}|}\right)$.
	
	\section{Generalized Results}\label{Ap4}

Upper bounds on CIR constitute our rationality assumptions for the agent. However, our results also rely on informational assumptions. Sections \ref{S4}, \ref{S5}, and \ref{S6} consider environments that differ primarily by how ``informed'' the agent appears, relative to the principal. In all three cases, however, we require the agent to be at least as informed as the principal. What is the principal's information? Recall that our mechanisms $\sigma^*$ will be forecast mechanisms \eqref{D2}. A calibrated learning algorithm -- which we specify later on -- will produce a sequence of forecasts $\pi_1,\ldots,\pi_T$. It is possible that these forecasts will become correlated with the state, e.g. if there is a trend in the data. We do not rule this out; however, if our forecasts inadvertently pick up useful information, this information should be available to the agent as well (either implicitly or because we publish $\pi_t$ along with $p_t$).

The notion of forecastwise regret (and forecastwise CIR) formalizes what we mean by the principal's ``information'' being available to the agent. The agent's benchmark includes the principal's forecast as additional context. Formally, define the forecast space $\mathcal{F}=\Delta(\mathcal{Y})$. Fix a small constant $\delta_\mathcal{F}>0$ and consider a finite partition $S_\mathcal{F}$ of $\mathcal{F}$ where $\pi,\tilde{\pi}\in F\in S_\mathcal{F}$ implies $d_\infty\left(\pi,\tilde{\pi}\right)\leq\delta_{\mathcal{F}}$. Let $\mathcal{F}_1\subseteq\mathcal{F}$ contain a single distribution $\pi\in F$ for every $F\in S_\mathcal{F}$. 

\begin{defn}
	Let the \textit{information partition} combine the forecast and CIR context, i.e.
	$$
	\mathcal{I}=S_{\mathcal{F}}\times\left(S_{\mathcal{R}}\right)^\Sigma
	$$
	and let the information $I_t\in\mathcal{I}$ in period $t$ be the unique set that satisfies
	$$
	\left(\pi_t,r_t^*,(r^p_t)_{p\in\mathcal{P}_0}\right)\in I_t
	$$
\end{defn}


\begin{defn}[FCIR]
	The agent's \emph{forecastwise CIR} relative to a modification rule $h:\mathcal{I}\to\mathcal{R}$ is
	\[
	\fcir{h}=\frac{1}{T}\sum_{t=1}^T\left(U(h(I_t),p_t,y_t)-U(r_t,p_t,y_t)\right)
	\]
	The FCIR relative to the best-in-hindsight modification rule is
	$
	\fcirs=\max_{h:\mathcal{I}\to\mathcal{R}}\fcir{h}
	$.
\end{defn}

To state our assumption, we need to define a forecastwise version of ER, just as we defined a forecastwise version of CIR at the end of section \ref{S3}. Let the forecast context $F_t\in S_{\mathcal{F}}$ in period $t$ be the unique set that satisfies $\pi_t\in F_t$.

\begin{defn}[FER]
	The agent's \emph{forecastwise external regret} relative to a strategy $h:S_{\mathcal{F}}\to\mathcal{R}$ is
	\[
	\fer{h}=\frac{1}{T}\sum_{t=1}^T\left(U(h(F_t),p_t,y_t)-U(r_t,p_t,y_t)\right)
	\]
	The FER relative to the best-in-hindsight strategy is
	$
	\fers=\max_{h:\mathcal{F}\to\mathcal{R}}\fer{h}
	$.
\end{defn}

\begin{theorem}\label{T4}
	Assume regularity \eqref{A1} and $\epsilon$-bounded FCIR \eqref{A4}. There exists a nonresponsive mechanism $\sigma^*$ parameterized by the agent's learner $L$ and a constant $\bar{\epsilon}>0$ such that
	\begin{enumerate}
		\item The principal's regret is bounded, i.e.
		\begin{align*}
			\ex[\sigma^*]{\pregs}
			\leq&
			\frac{1}{T}\sum_{I\in\mathcal{I}}n_I\Delta(\hat{\pi}_I,\bar\epsilon)
			+\frac{1}{\bar{\epsilon}}\left(O(\epsilon)
			+\tilde{O}\left(T^{-1/4}\sqrt{|\mathcal{Y}||S_{\mathcal{F}}||\mathcal{S}_R|^{(|\Sigma_0|+|S_\mathcal{P}|)/2}}\right)
			+O\left(\delta_{\mathcal{F}}^{1/2}\right)+O(\delta_{\mathcal{P}})\right)
		\end{align*}
	\end{enumerate}
\end{theorem}

\begin{assume}[Alignment]\label{A10}
	The stage game is $(\epsilon,M_1,M_2)$-aligned if, for all signals $\gamma$,
	\[
	\underbrace{\left(\phi_{p^*(\pi,\epsilon)}(\pi,\epsilon)-\alpha_{p^*(\pi,\epsilon)}(\pi,\gamma,\epsilon)\right)}_\text{maximum downside of $\gamma$ for the principal}
	\leq
	M_1\underbrace{\max_{r,r_J\in\mathcal{R}}\ex[y\sim\pi]{\ex[J\sim\tilde\gamma(\cdot,y)]{U(r_J,p^*(\pi,\epsilon),y)}-U(r,p^*(\pi,\epsilon),y)}}_\text{usefulness of $\gamma$ to the agent}+M_2
	\]
	and, for all policies $p\in\mathcal{P}_0$,
	\[
	\underbrace{\left(\beta_p(\pi,\gamma,\epsilon)-\phi_{p^*(\pi,\epsilon)}(\pi,\epsilon)\right)}_\text{maximum upside of $\gamma$ for the principal}
	\leq
	M_1\underbrace{\max_{r,r_J\in\mathcal{R}}\ex[y\sim\pi]{\ex[J\sim\tilde\gamma(\cdot,y)]{U(r_J,p,y)}-U(r,p,y)}}_\text{usefulness of $\gamma$ to the agent}+M_2
	\]
\end{assume}

\begin{theorem}\label{T5}
	Assume regularity \eqref{A1}, $\epsilon$-bounded FCIR \eqref{A4}, $\tilde{\epsilon}$-lower-bounded FER \eqref{A5}, and ($\bar{\epsilon},M_1,M_2$)-alignment \eqref{A10}. There exists a nonresponsive mechanism $\sigma^*$ parameterized by $\bar{\epsilon}$ such that
	\begin{enumerate}
		\item The principal's regret is bounded, i.e.
		\begin{align*}
			\ex[\sigma^*]{\pregs}
			\leq&
			\frac{1}{T}\sum_{F\in S_\mathcal{F}}n_F\Delta(\hat{\pi}_F,\bar\epsilon)
			+\frac{1}{\bar{\epsilon}}\left(O(\epsilon)+\tilde{O}\left(T^{-1/4}\sqrt{|\mathcal{Y}||S_{\mathcal{F}}|}\right)+O\left(\sqrt{\delta_{\mathcal{F}}}\right)+O(\delta_{\mathcal{P}})\right)\\
			&+O(\tilde{\epsilon})
			+M_1\left(O(\tilde{\epsilon})+O(\epsilon)+\tilde{O}\left(T^{-1/4}\sqrt{|\mathcal{Y}||S_{\mathcal{F}}|}\right)+O\left(\sqrt{\delta_{\mathcal{F}}}\right)+O(\delta_{\mathcal{P}})\right)
			+O(M_2)
		\end{align*}
	\end{enumerate}
\end{theorem}

\begin{theorem}\label{T6}
	Assume regularity \eqref{A1} and $\epsilon$-bounded FCIR \eqref{A4}. There exists a nonresponsive mechanism $\sigma^*$ parameterized by a constant $\bar{\epsilon}>0$ such that
	\begin{enumerate}
		\item The principal's regret is bounded, i.e.
		\begin{align*}
			\ex[\sigma^*]{\pregs}
			\leq&
			\frac{1}{T}\sum_{F\in S_\mathcal{F}}n_F\Delta(\hat{\pi}_F,\bar\epsilon)
			+\frac{1}{\bar{\epsilon}}\left(O(\epsilon)+\tilde{O}\left(T^{-1/4}\sqrt{|\mathcal{Y}||S_{\mathcal{F}}|}\right)+O\left(\delta_{\mathcal{F}}^{1/2}\right)+O(\delta_{\mathcal{P}})\right)
		\end{align*}
	\end{enumerate}
\end{theorem}

	\section{Omitted Proofs}\label{Ap5}
	
	\subsection{Proof of Propositions \ref{???} and \ref{???}}

Recall that a policy $p_t$ in period $t$ can affect the agent's behavior $\mu_\tau$ in period $\tau>t$. This raises the prospect that a mistake today can cause irreversible damage to the principal's average utility. By definition, the principal will regret that mistake. This would make the principal's problem infeasible, in that he cannot guarantee low regret for himself.

Generally-speaking, regret bounds can bypass this problem if they restrict how much the agent's response $r_t$ depends on the policy history $p_{1:t-1}$. This is reasonable a priori, since the policy history $p_{1:t-1}$ appears irrelevant to the agent's problem. It neither affects nor predicts the state $y_t$, except through its dependence on $y_{1:t-1}$. It is not needed as a predictor of the policy $p_t$ because the agent observes $p_t$ before choosing a response.\footnote{For example, if the agent were Bayesian then there would be no dependence on $p_{1:t}$ at all. If the agent used the exponential weights algorithm then there would only be an indirect dependence, since that algorithm depends on the agent's historical payoffs and these, in turn, depend on the policy history.} For these reasons, it seems that the agent can only make herself worse off by allowing irrelevant variables like $p_{1:t-1}$ to affect her response $r_t$. This would be true if our notion of good performance overall had clear implications for behavior in each period, so that unecessary variation in behavior implies a departure from optimality. Unfortunately, there are various kinds of behavior that obtain low ER. The agent can easily switch between these behaviors while still satisfying no-ER. In the process, she can cause substantial benefit or harm to the principal.

To clarify the problem, we present several examples of learners that we regard as pathological. These are implicit counterexamples to the proposition that no-ER constraints are sufficient for no-regret mechanism design.

Our counterexamples are closely related to the pathological phenomenon of ``superefficiency'' in statistics. Suppose we are trying to estimate the mean $\theta$ of the normal random variable $X\sim N(\theta,1)$, given an i.i.d. random sample $X_1,\ldots,X_n$. Our objective is to minimize the mean square error, but this depends on the parameter $\theta$. A typical solution is the maximum likelihood estimator (MLE), which in this case outputs the sample mean $n^{-1}\sum_{i=1}^n X_i$. For reasons that are unimportant to our discussion, MLE is considered ``efficient''. However, it is easy to find an estimator that outperforms MLE. For example, a wild-ass guess (WAG) ignores the data and outputs $\theta^*$. If it happens to be the case that $\theta=\theta^*$ then this estimator is optimal.

In the following example, we construct a learner that alternates between a WAG-like predictor and a MLE-like predictor depending on a seemingly irrelevant choice by the principal.

\begin{example}[Selective Superefficiency]\label{Ex1}
	Consider a learner $L$ that is capable --  either by ingenuity or dumb luck -- of predicting the state sequence $y_{1:T}$ perfectly. However, the learner uses this ability only selectively, depending on the state $y_1$ and policy $p_1$ in the first period. Despite this seemingly irrational behavior, the learner satisfies vanishing external regret.
	
	Let $P\subsetneq\mathcal{P}$ be a nonempty subset of policies. Let $Y\subsetneq \mathcal{Y}$ be a nonempty subset of states. Let $r*$ be the best-in-hindsight response by time $T$. That is, consider some $r^*\in\mathcal{F}$ that happens to be best-in-hindsight given the realized state sequence $y_{1:T}$ but will not be best-in-hindsight uniformly over all state sequences. Given $y_{1:T}$, define the learner $L$ as follows:
	\begin{enumerate}
		\item If $y_1\in Y$ and $p_1\in P$ then use response $r^*$
		
		\item If $y_1\in Y$ and $p_1\notin P$ then use the response that happens to be optimal given $y_t$.
		
		\item If $y_1\notin Y$ and $p_1\in P$ then use the response that happens to be optimal given $y_t$.
		
		\item If $y_1\notin Y$ and $p_1\notin P$ then use response $r^*$
	\end{enumerate}
	In cases 1 and 4, the learner follows the best-in-hindsight response and therefore achieves zero regret. In cases 2 and 3, the learner acts optimally ex post and therefore achieves non-positive regret.
	
	Nonetheless, no mechanism can guarantee no-regret for the principal. Suppose the mechanism chooses $p_1\in P$. If it turns out that $y_1\in Y$ then the agent will follow $\pi^*$. Otherwise, the agent will be superefficient. Were the mechanism to deviate to $p_1\notin P$, the situation would be reversed. These constitute permanent changes in the agent's behavior. Suppose one type of behavior is ``better'' for the principal than another. It is always possible in hindsight that the mechanism's first-period policy was the one that led to the ``worse'' type of behavior.
\end{example}

Can further assumptions rule out this kind of behavior? Again, consider the analogy with statistics. The WAG estimator -- always predict $\theta^*$ -- will perform very poorly in the counterfactual world where $\theta^*\neq\theta$. Formally, this estimator is not ``consistent''. Similarly, the learner from example \ref{Ex1} does not guarantee vanishing average external regret under counterfactual state sequences. This reflects a peculiar unresponsiveness to the data.

Unfortunately, imposing consistency or no-regret on all sequences does not rule out these kinds of pathologies. Consider Hodge's (superefficient) estimator, which outputs $\theta^*$ unless there is sufficient evidence that $\theta\neq\theta^*$. In that case, it outputs the sample mean. If ``sufficient evidence'' is defined carefully, this estimator will outperform MLE when $\theta=\theta^*$ and will asympotically match MLE otherwise. We can patch up example \ref{Ex1} in a similar way.

\begin{example}[Selective Superefficiency, revised]\label{Ex2}
	We want to modify the learner in example \ref{Ex1} to ensure no-regret on all counterfactual state sequences $\tilde{y}_{1:T}$. This is straightforward. In period $t+1$, if the history $\tilde{y}_{1:t}$ matches the presumed sequence $y_{1:t}$ exactly, then proceed as before. Otherwise, follow any no-ER algorithm. This guarantees vanishing average regret as long as the regret from the first period $t$ where $\tilde{y}_{1:t}\neq y_{1:t}$ is bounded -- which it is, since our utility functions are bounded.
	
	Therefore, for any sequence $y_{1:t}$ there is a learner that satisfies no-regret on all sequences, but exhibits the pathological behavior from example \ref{Ex1} on the realized sequence.
\end{example}

Statisticians deal with superefficiency by arguing that it generically fails to occur. That is, any alternative estimator will weakly underperform MLE on Lebesgue-almost all values $\theta$. For example, we can view Hodge's estimator as asymptotically equivalent to MLE whenever $\theta\neq\theta^*$. In our setting, attempting this argument would necessitate a definition of genericity for sequences of states. While we can provide various definitions, none seem especially compelling.\footnote{For example, one could assign equal measure to each sequence $y_{1:T}$ in the set $\mathcal{Y}^T$ of all possible sequences. By this measure, the measure of any constant sequence $(y,\ldots,y)$ would converge to zero as $T\to\infty$. Yet it does not seem unreasonable a priori that the world should persist in a fixed state. Alternatively, one could assign equal measure to all permutations of a given sequence $y_{1:T}$. This would effectively return us to an i.i.d. setting.}

Another natural restriction to impose is that the learner should not outperform the experts. That is, given the sequence $y_{1:T}$, we only consider learners whose regret at period $T$ is non-negative. Clearly, this rules out the learners in examples \ref{Ex1} and \ref{Ex2}, which may obtain negative regret in the sequences where it predicts the state perfectly and acts on that information. Unfortunately, this does not rule out the broader phenomenon, as the following example illustrates.

\begin{example}[Selective Superinefficiency]\label{Ex3}
	
	As in example \ref{Ex1}, define a learner $L$ that appears capable of predicting the state sequence $y_{1:T}$ perfectly. This learner will continue to use this ability selectively. Moreover, when the learner uses this ability, she does not always use it to her advantage. Instead, with probability $1-q$ she uses it to her own disadvantage. When $q$ is chosen correctly, the learner satisfies zero regret for all mechanisms.
	
	Let $P\subsetneq \mathcal{P}$ be a nonempty subset of policies. Let $Y\subsetneq\mathcal{Y}$ be a nonempty subset of states. Let $r^*$ be the best-in-hindsight response by time $T$. Let $r^\dagger_t$ be the response that happens to be optimal for $y_t$. Let $\tilde{r}_t$ be the response that minimizes the agent's utility when the state is $y_t$. Given $y_{1:T}$, define the learner $L$ as follows:
	\begin{enumerate}
		\item If $y_1\in Y$ and $p_1\in P$ then use response $r^*$
		
		\item If $y_1\in Y$ and $p_1\notin P$ then use $r^\dagger_t$ with probability $q$ and $\tilde{r}_t$ with probability $1-q$
		
		\item If $y_1\notin Y$ and $p_1\in P$ then use $r^\dagger_t$ with probability $q$ and $\tilde{r}_t$ with probability $1-q$
		
		\item If $y_1\notin Y$ and $p_1\notin P$ then use prior $\pi^*$
	\end{enumerate}
	In cases 1 and 4, the learner follows the best-in-hindsight prior and therefore achieves zero regret. In cases 2 and 3, as long as $\tilde{r}_t$ underperforms $r^*$ in every period $t$, by continuity there exists a probability $q$ such that the agent achieves zero regret.
	
	This learner now satisfies both an upper bound and a lower bound on regret. Nonetheless, our difficulties remain. Just as in example \ref{Ex1}, the first-period policy can cause permanent changes in the agent's behavior. It is always possible in hindsight that the mechanism's first-period policy was the one that led to the ``worse'' type of behavior.
\end{example}

We can use this example to prove proposition \ref{???} (proposition \ref{???} is a corollary). In the Bayesian persuasion example, let $y_{2:T}$ be drawn i.i.d. where the defendant is guilty with probability $q=0.5-\epsilon$ for a very small $\epsilon>0$. If the principal chooses $p_1$ correctly, he can persuade the agent to convict with probability near one. Otherwise, the agent convicts with probability near 0.5.

In the contract theory example, let $y_{2:T}$ be be drawn i.i.d. from some distribution where the principal would find it optimal to pay the agent in the stage game, but both states occur with positive probability. If the principal chooses $p_1$ correctly, he can pay the agent her cost of effort and achieve the first-best outcome (agent works iff working is effective). Otherwise, the principal has to compensate the agent for her cost of effort in states where working is ineffective.

The fundamental problem with the learners in examples \ref{Ex1}, \ref{Ex2}, \ref{Ex3} is not that they are well-informed. After all, in some settings we might reasonably expect the agent to be better informed than the analyst. The problem is that they fail to consistently and fully exploit the private information that they clearly possess. Bounds on counterfactual internal regret capture this failure to exploit information and rule out these kinds of pathological behaviors.

\begin{example}\label{Ex4}
	Returning to example \ref{Ex1}, consider two constant mechanisms $\sigma^p$ and $\sigma^{\tilde p}$ where $p\in P$ and $\tilde{p}\notin P$. Regardless of the state sequence $y_{1:T}$, exactly one of these mechanisms (say $\sigma^p$) will cause the agent to predict the state perfectly while the other will cause the agent to follow the best-in-hindsight prior. The agent's behavior $r^p$ following $\sigma^p$ will differ across periods $y_t,y_\tau$ if and only if $y_t\neq y_\tau$, while the behavior $r^{\tilde p}$ following $\sigma^{\tilde{p}}$ remains constant throughout. The vector $(r^p,r^{\tilde{p}})$ will therefore differ across periods $y_t,y_\tau$ if and only if $y_t\neq y_\tau$.
	
	If we require the agent to have no-contextual regret where the context is $(r^p,r^{\tilde{p}})$, it is equivalent to requiring her to predict the state perfectly even if the principal uses $\sigma^{\tilde p}$. This is essentially the context used to define CIR. The learner guarantees no-CIR under mechanism $\sigma^{p}$, because it predicts the state perfectly. However, it does not predict the state perfectly under mechanism $\sigma^{\tilde p}$, so in this case the agent accumulates CIR.
\end{example}
	
	\subsection{Proof of Lemmas \ref{???} and \ref{???} and Additional Results}

The lemmas in this section will be used repeatedly in the proofs of theorems \ref{T1}, \ref{T2}, and \ref{T3}.

\subsubsection{Proof of Lemmas \ref{???} and \ref{???}}

Lemma \ref{???} states that for any policy $p$, information structure $\gamma$, constants $\epsilon,\tilde{\epsilon}>0$, and distribution $\pi$, we have
\[
\alpha_p(\pi,\gamma,\epsilon+\tilde\epsilon)\geq\alpha_p(\pi,\gamma,\epsilon)-\frac{\tilde\epsilon}{\epsilon}
\quad\mathrm{and}\quad
\beta_p(\pi,\gamma,\epsilon+\tilde\epsilon)\leq\beta_p(\pi,\gamma,\epsilon)+\frac{\tilde\epsilon}{\epsilon}
\]
Note that lemma \ref{???} is just a special case where the information structure $\gamma$ is uninformative. To prove this, define
\[
B(\pi,\gamma,\epsilon)=\left\{\mu_J\in\Delta(\mathcal{R})\mid\epsilon\geq\max_{\tilde{r}_J\in\mathcal{R}}\ex[y\sim\pi]{\ex[J\sim\gamma(\cdot,y)]{U(\tilde{r}_J,p,y)-\ex[r\sim\mu_J]{U(r,p,y)}}}\right\}
\]
and recall that
\[
\alpha_p(\pi,\gamma,\epsilon)=\min_{\mu_J\in B(\pi,\gamma,\epsilon)}\ex[y\sim\pi]{\ex[J\sim\gamma(\cdot,y)]{\ex[r\sim\mu_J]{V(r,p,y)}}}
\]
Note that $\alpha_p(\pi,\gamma,\epsilon)$ is decreasing and convex in $\epsilon$. Convexity follows from the fact that $\mu_J\in B(\pi,\gamma,\epsilon)$ and $\tilde\mu_J\in B(\pi,\gamma,\tilde\epsilon)$ implies  $\lambda\mu_J+(1-\lambda)\tilde\mu_J\in B(\pi,\gamma,\lambda\epsilon+(1-\lambda)\tilde\epsilon)$. Therefore,
\[
\alpha_p(\pi,\gamma,\lambda\epsilon+(1-\lambda)\tilde\epsilon)\leq\lambda\alpha_p(\pi,\gamma,\epsilon)+(1-\lambda)\alpha_p(\pi,\gamma,\tilde\epsilon)
\]
Consider any supporting line of $\alpha_p$ at $\epsilon$. It is bounded above by $\alpha_p$, by definition. Therefore, its slope is at most
\[
\frac{\alpha_p(\pi,\gamma,0)-\alpha_p(\pi,\gamma,\epsilon)}{\epsilon}
\leq\frac{1}{\epsilon}
\]
since $\alpha_p$ is bounded in the unit interval by our regularity assumption. Therefore, the supporting line will underestimate $\alpha_p(\pi,\gamma,\epsilon+\tilde{\epsilon})$ by at most $\tilde{\epsilon}/\epsilon$ and at least zero. This implies our bound. The argument for $\beta_p$ is analogous after we observe that it is increasing and concave in $\epsilon$.

\subsubsection{Bounds for Misspecified Distributions}

The following lemma states that the principal's worst-case utility $\alpha_p$ is not too sensitive to changes in the distribution, for any fixed policy $p$.

\begin{lemma}\label{???1}
	For any policy $p$, information structure $\gamma$, constant $\epsilon>0$, and distributions $\pi,\tilde{\pi}$, we have
	\[
	\alpha_p(\pi,\gamma,\epsilon)\geq\alpha_p(\tilde\pi,\gamma,\epsilon)-\frac{2d_1(\pi,\tilde{\pi})}{\epsilon}-d_1(\pi,\tilde{\pi})
	\]
\end{lemma}
\begin{proof}
	Note that $B(\pi,\gamma,\epsilon)\subseteq B(\tilde{\pi},\gamma,\epsilon+2d_1(\pi,\tilde{\pi}))$ since 
	(1) for any $\tilde{r}_J$,
		\[
		\ex[y\sim\pi]{\ex[J\sim\gamma(\cdot,y)]{U(\tilde{r}_J,p,y)}}
		\geq\ex[y\sim\tilde\pi]{\ex[J\sim\gamma(\cdot,y)]{U(\tilde{r}_J,p,y)}}-d_1(\pi,\tilde{\pi})
		\]
	and (2), for any $\mu_J$,
		\[
		\ex[y\sim\pi]{\ex[J\sim\gamma(\cdot,y)]{\ex[r\sim\mu^*_J]{U(r,p,y)}}}
		\leq\ex[y\sim\tilde\pi]{\ex[J\sim\gamma(\cdot,y)]{\ex[r\sim\mu^*_J]{U(r,p,y)}}}+d_1(\pi,\tilde{\pi})
		\]
	Therefore, 
	\begin{align*}
	\alpha_p(\pi,\gamma,\epsilon)
	&\geq\min_{\mu_J\in B(\tilde\pi,\gamma,\epsilon+2d_1(\pi,\tilde{\pi}))}\ex[y\sim\pi]{\ex[J\sim\gamma(\cdot,y)]{\ex[r\sim\mu_J]{V(r,p,y)}}}\\
	&\geq\min_{\mu_J\in B(\tilde\pi,\gamma,\epsilon+2d_1(\pi,\tilde{\pi}))}\ex[y\sim\tilde\pi]{\ex[J\sim\gamma(\cdot,y)]{\ex[r\sim\mu_J]{V(r,p,y)}}}-d_1(\pi,\tilde{\pi})\\
	&=\alpha_p(\tilde\pi,\gamma,\epsilon+2d_1(\pi,\tilde{\pi}))-d_1(\pi,\tilde{\pi})\\
	&\geq\alpha_p(\tilde\pi,\gamma,\epsilon)-\frac{2d_1(\pi,\tilde{\pi})}{\epsilon}-d_1(\pi,\tilde{\pi})
	\end{align*}
\end{proof}

The following lemma states that the $\epsilon$-robust policy for a distribution $\tilde{\pi}$ that is near the true distribution $\pi$ will perform almost as well as the $\epsilon$-robust policy for the true distribution $\pi$.

\begin{lemma}\label{???2}
	For any $\epsilon>0$ and distributions $\pi,\tilde{\pi}$, we have
	\[
	\alpha_{p^*(\tilde\pi,\epsilon)}(\pi,\epsilon)\geq\alpha_{p^*(\pi,\epsilon)}(\pi,\epsilon)-\frac{4d_1(\pi,\tilde{\pi})}{\epsilon}-2d_1(\pi,\tilde{\pi})
	\]
\end{lemma}
\begin{proof}
	First, observe that
	\begin{align*}
	\alpha_{p^*(\pi,\epsilon)}(\pi,\epsilon)
	&\leq\alpha_{p^*(\pi,\epsilon)}(\tilde\pi,\epsilon)+\frac{2d_1(\pi,\tilde{\pi})}{\epsilon}+d_1(\pi,\tilde{\pi})\\
	&\leq\alpha_{p^*(\tilde\pi,\epsilon)}(\tilde\pi,\epsilon)+\frac{2d_1(\pi,\tilde{\pi})}{\epsilon}+d_1(\pi,\tilde{\pi})
	\end{align*}
	Next, observe that
	\[
	\alpha_{p^*(\tilde\pi,\epsilon)}(\tilde\pi,\epsilon)\leq\alpha_{p^*(\tilde\pi,\epsilon)}(\pi,\epsilon)+\frac{2d_1(\pi,\tilde{\pi})}{\epsilon}+d_1(\pi,\tilde{\pi})
	\]
	Collapse these inequalities to obtain the desired result.
\end{proof}

The following lemma states that the $\epsilon$-informationally-robust policy for a distribution $\tilde{\pi}$ that is near the true distribution $\pi$ will provide a similar guarantee against the worst-case information structure $\gamma$ as the $\epsilon$-informationally-robust policy for the true distribution $\pi$.

\begin{lemma}\label{???3}
	For any $\epsilon>0$ and distributions $\pi,\tilde{\pi}$, we have
	\[
	\inf_\gamma\alpha_{p^\dagger(\pi,\epsilon)}(\pi,\gamma,\epsilon)\leq\inf_\gamma\alpha_{p^\dagger(\tilde\pi,\epsilon)}(\pi,\gamma,\epsilon)+\frac{4d_1(\pi,\tilde{\pi})}{\epsilon}+2d_1(\pi,\tilde{\pi})
	\]
\end{lemma}
\begin{proof}
	First, observe that
	\[
	\alpha_{p^\dagger(\pi,\epsilon)}(\pi,\gamma,\epsilon)\leq\alpha_{p^\dagger(\pi,\epsilon)}(\tilde\pi,\gamma,\epsilon)+\frac{2d_1(\pi,\tilde{\pi})}{\epsilon}+d_1(\pi,\tilde{\pi})
	\]
	which implies
	\begin{align*}
	\inf_\gamma\alpha_{p^\dagger(\pi,\epsilon)}(\pi,\gamma,\epsilon)
	&\leq\inf_\gamma\alpha_{p^\dagger(\pi,\epsilon)}(\tilde\pi,\gamma,\epsilon)+\frac{2d_1(\pi,\tilde{\pi})}{\epsilon}+d_1(\pi,\tilde{\pi})\\
	&\leq\inf_\gamma\alpha_{p^\dagger(\tilde\pi,\epsilon)}(\tilde\pi,\gamma,\epsilon)+\frac{2d_1(\pi,\tilde{\pi})}{\epsilon}+d_1(\pi,\tilde{\pi})
	\end{align*}
	Next, observe that
	\[
	\alpha_{p^\dagger(\tilde\pi,\epsilon)}(\tilde\pi,\gamma,\epsilon)
	\leq\alpha_{p^\dagger(\tilde\pi,\epsilon)}(\pi,\gamma,\epsilon)+\frac{2d_1(\pi,\tilde{\pi})}{\epsilon}+d_1(\pi,\tilde{\pi})
	\]
	which implies
	\[
	\inf_\gamma\alpha_{p^\dagger(\tilde\pi,\epsilon)}(\tilde\pi,\gamma,\epsilon)
	\leq\inf_\gamma\alpha_{p^\dagger(\tilde\pi,\epsilon)}(\pi,\gamma,\epsilon)+\frac{2d_1(\pi,\tilde{\pi})}{\epsilon}+d_1(\pi,\tilde{\pi})
	\]
	Collapse these inequalities to obtain the desired result.
\end{proof}
	
	\subsection{Proof of Theorem \ref{T4}}

Assume access to a forecast $\pi_t\in\Delta(\mathcal{Y})$ for every period $t$. We will define this later. In period $t$, the mechanism computes the policy $p^*(\pi_t,\bar{\epsilon})$ that maximizes the worst-case payoff in the $\bar{\epsilon}$-robust stage game, treating the forecast $\pi_t$ as the common prior. That is,
\[
p^*(\pi_t,\bar{\epsilon})\in\arg\max_{p\in\mathcal{P}}\alpha_p(\pi_t,\bar{\epsilon})
\]
The mechanism chooses $p_t$ as follows. Let $P$ be the unique policy context that includes the policy $p^*(\pi_t,\bar{\epsilon})\in P$. Let $p_t=p_P$, where $p_P\in\mathcal{P}_1$ is the representative element of $P$.

We will refer to the average regret accumulated in each forecast context $F$, i.e.
\[
\epsilon_F=\max_{r\in\mathcal{R}}\frac{1}{n_F}\sum_{t\in F}\left(U(r,p_F,y_t)-U(r_t,p_F,y_t)\right)
\]
where $p_F=p_P$ for the unique policy context $P$ associated with forecast context $F$. We will also refer to the average regret accumulated in each information context $I\in\mathcal{I}$, i.e.
\[
\epsilon_I=\max_{r\in\mathcal{R}}\frac{1}{n_I}\sum_{t\in I}\left(U(r,p_I,y_t)-U(r_t,p_I,y_t)\right)
\]
where $p_I=p_F$ for the unique forecast context $F$ associated with information $I$. Note that
\[
\epsilon_F
=\frac{1}{n_F}\sum_{I\in\mathcal{I}_F}n_I\epsilon_I
=\frac{1}{n_F}\sum_{I\in\mathcal{I}_F}\max_{r\in\mathcal{R}}\frac{1}{n_I}\sum_{t\in I}\left(U(r,p_F,y_t)-U(r_t,p_F,y_t)\right)
\]
The next two lemmas imply an upper bound on the principal's regret in terms of the quantity
\[
\iota\geq\frac{1}{T}\sum_{t=1}^Td_1(\pi_t,\hat{\pi}_I)
\]
that measures the discrepancy between the forecast $\pi_t$ and the empirical distribution $\hat{\pi}_I$ conditioned on the agent's information $I$. Lemma \ref{L2} is a lower bound on the principal's payoff under $\sigma^*$. Lemma \ref{L3} is an upper bound on the his payoff under any constant $\sigma^p\in\Sigma_0$.

\begin{lemma}\label{L2}
	Suppose the principal runs $\sigma^*$. Then
	\[
	\frac{1}{T}\sum_{I\in\mathcal{I}}\sum_{t\in I}
	V(r_t,p_t,y_t)
	\geq\max_{p\in\mathcal{P}}\alpha_p(\hat{\pi}_I,\bar\epsilon)
	-\left(\frac{\epsilon+4\iota+K^U_{\mathcal{R}}\delta_{\mathcal{R}}+2K^U_{\mathcal{P}}\delta_{\mathcal{P}}}{\bar{\epsilon}}\right)
	-\left(2\iota+K^V_{\mathcal{R}}\delta_{\mathcal{R}}+K^V_{\mathcal{P}}\delta_{\mathcal{P}}\right)
	\]
\end{lemma}
\begin{proof}
	Let $r_I$ be a representative element in the response context $R$ associated with information $I$ under mechanism $\sigma^*$. By regularity,
	\begin{align*}
		\epsilon_I
		&\geq\max_{r\in\mathcal{R}}\frac{1}{n_I}\sum_{t\in I}\left(U(r,p_I,y_t)-U(r_I,p_I,y_t)\right)-K^U_{\mathcal{R}}\delta_{\mathcal{R}}\\
		&=\max_{r\in\mathcal{R}}\ex[y\sim\hat{\pi}_I]{U(r,p_I,y)-U(r_I,p_I,y)}-K^U_{\mathcal{R}}\delta_{\mathcal{R}}
	\end{align*}
	It follows, by regularity and definition of $\alpha$, that
	\begin{align*}
		\frac{1}{n_I}\sum_{t\in I}V(r_t,p_I,y_t)
		&\geq\ex[y\sim\hat{\pi}_I]{V(r_I,p_I,y)}-K^V_{\mathcal{R}}\delta_{\mathcal{R}}\\
		&\geq\alpha_{p_I}(\hat{\pi}_I,\epsilon_I+K^U_{\mathcal{R}}\delta_{\mathcal{R}})-K^V_{\mathcal{R}}\delta_{\mathcal{R}}
	\end{align*}
	Summing over information $I\in\mathcal{I}$ and using lemma \ref{L1}, we obtain
	\begin{align*}
		\frac{1}{T}\sum_{I\in\mathcal{I}}\sum_{t\in I}
		V(r_t,p_I,y_t)
		&\geq\frac{1}{T}\sum_{I\in\mathcal{I}}\sum_{t\in I}
		\alpha_{p_I}(\hat{\pi}_I,\epsilon_I+K^U_{\mathcal{R}}\delta_{\mathcal{R}})-K^V_{\mathcal{R}}\delta_{\mathcal{R}}\\
		&\geq\frac{1}{T}\sum_{I\in\mathcal{I}}\sum_{t\in I}\alpha_{p^*(\pi_t,\bar{\epsilon})}(\hat{\pi}_I,\epsilon_I+K^U_{\mathcal{R}}\delta_{\mathcal{R}}+2K^U_{\mathcal{P}}\delta_{\mathcal{P}})-K^V_{\mathcal{R}}\delta_{\mathcal{R}}-K^V_{\mathcal{P}}\delta_{\mathcal{P}}\\
		&\geq\frac{1}{T}\sum_{I\in\mathcal{I}}\sum_{t\in I}\left(\alpha_{p^*(\pi_t,\bar{\epsilon})}(\hat{\pi}_I,\bar\epsilon)-\frac{\epsilon_I+K^U_{\mathcal{R}}\delta_{\mathcal{R}}+2K^U_{\mathcal{P}}\delta_{\mathcal{P}}}{\bar{\epsilon}}\right)-K^V_{\mathcal{R}}\delta_{\mathcal{R}}-K^V_{\mathcal{P}}\delta_{\mathcal{P}}\\
		&\geq\frac{1}{T}\sum_{I\in\mathcal{I}}\sum_{t\in I}
		\alpha_{p^*(\pi_t,\bar{\epsilon})}(\hat{\pi}_I,\bar\epsilon)-\left(\frac{\epsilon+K^U_{\mathcal{R}}\delta_{\mathcal{R}}+2K^U_{\mathcal{P}}\delta_{\mathcal{P}}}{\bar{\epsilon}}\right)-K^V_{\mathcal{R}}\delta_{\mathcal{R}}-K^V_{\mathcal{P}}\delta_{\mathcal{P}}
	\end{align*}
	Focus on the first term, i.e.
	\begin{align*}
		\frac{1}{T}\sum_{I\in\mathcal{I}}\sum_{t\in I}
		\alpha_{p^*(\pi_t,\bar{\epsilon})}(\hat{\pi}_I,\bar\epsilon)
		&\geq\frac{1}{T}\sum_{I\in\mathcal{I}}\sum_{t\in I}
		\left(\alpha_{p^*(\hat\pi_I,\bar{\epsilon})}(\hat{\pi}_I,\bar\epsilon)-\frac{4d_1(\pi_t,\hat{\pi}_I)}{\bar{\epsilon}}-2d_1(\pi_t,\hat{\pi}_I)\right)\\
		&\geq\frac{1}{T}\sum_{I\in\mathcal{I}}n_I
		\alpha_{p^*(\hat\pi_I,\bar{\epsilon})}(\hat{\pi}_I,\bar\epsilon)-\frac{4\iota}{\bar{\epsilon}}-2\iota
	\end{align*}
	Collapsing these inequalities gives us the desired bound.
\end{proof}

\begin{lemma}\label{L3}
	Suppose the principal runs some constant mechanism $\sigma^p\in\Sigma_0$. Then
	\[
	\frac{1}{T}\sum_{t=1}^T
	V(r_t,p,y_t)
	\leq\frac{1}{T}\sum_{I\in\mathcal{I}}n_I
	\max_{\tilde{p}\in\mathcal{P}}\beta_{\tilde{p}}(\hat{\pi}_I,\bar\epsilon)
	+\left(\frac{\epsilon+K^U_{\mathcal{R}}\delta_{\mathcal{R}}}{\bar{\epsilon}}\right)
	+K^V_{\mathcal{R}}\delta_{\mathcal{R}}
	\]
\end{lemma}
\begin{proof}
	Let $r_I$ be a representative element in the response context $R$ associated with information $I$ under $\sigma^p$. By regularity,
	\begin{align*}
		\epsilon_I
		&\geq\max_{r\in\mathcal{R}}\frac{1}{n_I}\sum_{t\in I}\left(U(r,p_I,y_t)-U(r_I,p_I,y_t)\right)-K^U_{\mathcal{R}}\delta_{\mathcal{R}}\\
		&=\max_{r\in\mathcal{R}}\ex[y\sim\hat{\pi}_I]{U(r,p_I,y)-U(r_I,p_I,y)}-K^U_{\mathcal{R}}\delta_{\mathcal{R}}
	\end{align*}
	It follows, by regularity and definition of $\beta$, that
	\begin{align*}
		\frac{1}{n_I}\sum_{t\in I}V(r_t,p,y_t)
		&\leq\ex[y\sim\hat{\pi}_I]{V(r_I,p,y)}+K^V_{\mathcal{R}}\delta_{\mathcal{R}}\\
		&\leq\alpha_{p}(\hat{\pi}_I,\epsilon_I+K^U_{\mathcal{R}}\delta_{\mathcal{R}})+K^V_{\mathcal{R}}\delta_{\mathcal{R}}
	\end{align*}
	Summing over information $I\in\mathcal{I}$ and using lemma \ref{L1}, we obtain
	\begin{align*}
		\frac{1}{T}\sum_{I\in\mathcal{I}}\sum_{t\in I}
		V(r_t,p,y_t)
		&\leq\frac{1}{T}\sum_{I\in\mathcal{I}}n_I
		\beta_{p}(\hat{\pi}_I,\epsilon_I+K^U_{\mathcal{R}}\delta_{\mathcal{R}})+K^V_{\mathcal{R}}\delta_{\mathcal{R}}\\
		&\leq\frac{1}{T}\sum_{I\in\mathcal{I}}n_I
		\left(\beta_{p}(\hat{\pi}_I,\bar\epsilon)+\frac{\epsilon_I+K^U_{\mathcal{R}}\delta_{\mathcal{R}}}{\bar{\epsilon}}\right)
		+K^V_{\mathcal{R}}\delta_{\mathcal{R}}\\
		&\leq\frac{1}{T}\sum_{I\in\mathcal{I}}n_I
		\beta_{p}(\hat{\pi}_I,\bar\epsilon)+\left(\frac{\epsilon+K^U_{\mathcal{R}}\delta_{\mathcal{R}}}{\bar{\epsilon}}\right)+K^V_{\mathcal{R}}\delta_{\mathcal{R}}\\
		&\leq\frac{1}{T}\sum_{I\in\mathcal{I}}n_I
		\max_{\tilde{p}\in\mathcal{P}}\beta_{\tilde{p}}(\hat{\pi}_I,\bar\epsilon)+\left(\frac{\epsilon+K^U_{\mathcal{R}}\delta_{\mathcal{R}}}{\bar{\epsilon}}\right)+K^V_{\mathcal{R}}\delta_{\mathcal{R}}
	\end{align*}
	This is the desired bound.
\end{proof}

From these two lemmas, it immediately follows that
\[
\pregs
\leq\frac{1}{T}\sum_{I\in\mathcal{I}}n_I
\Delta(\hat{\pi}_I,\bar\epsilon)
+2\left(\frac{\epsilon+2\iota+K^U_{\mathcal{R}}\delta_{\mathcal{R}}+K^U_{\mathcal{P}}\delta_{\mathcal{P}}}{\bar{\epsilon}}\right)
+\left(2\iota+2K^V_{\mathcal{R}}\delta_{\mathcal{R}}+K^V_{\mathcal{P}}\delta_{\mathcal{P}}\right)
\]
Therefore, to bound the principal's regret, all that remains is to bound $\iota$.

Set $\mathcal{C}=S_{\mathcal{R}}^{|\mathcal{P}_0|+|\mathcal{P}_1|}$ where $C_t(\mathcal{P}_1)$ is the vector describing the agent's response contexts $R_t$ under policy history $p^*_{1:t-1}$ and policy choices $p_t\in\mathcal{P}_1$. Note that this is different from the behavior context that we used to define the agent's information, which refers to the response context $R_t$ under policy history $p^*_{1:t}$. Because we are currently designing the mechanism, we cannot refer to $p^*_t$ without attempting to solve a fixed point problem that may not have a solution.

We use the algorithm from appendix \ref{Ap2-0} to generate $\pi_t$, with a modification: run it separately for each context $C_t$. Adapting equation \eqref{E4}, we obtain
\begin{align*}
\ex[\sigma^*]{\frac{1}{T}\sum_{C\in\mathcal{C}}\sum_{F\in\mathcal{F}}n_{F,C}d_1(\pi_t,\hat{\pi}_C)}
&\leq
\frac{1}{T}\sum_{C\in\mathcal{C}}n_C\sqrt{|\mathcal{Y}||\mathcal{F}|\sqrt{\frac{2\log |\mathcal{F}|}{n_C}}+2|\mathcal{Y}|\delta_{\mathcal{F}}}\\
&\leq\frac{1}{T}\sum_{C\in\mathcal{C}}n_C^{3/4}\sqrt{|\mathcal{Y}||\mathcal{F}|\sqrt{2\log |\mathcal{F}|}}+\sqrt{2|\mathcal{Y}|\delta_{\mathcal{F}}}\\
&\leq\frac{1}{T}\sum_{C\in\mathcal{C}}\left(\frac{T}{|\mathcal{C}|}\right)^{3/4}\sqrt{|\mathcal{Y}||\mathcal{F}|\sqrt{2\log |\mathcal{F}|}}+\sqrt{2|\mathcal{Y}|\delta_{\mathcal{F}}}\\
&=\left(\frac{|\mathcal{C}|}{T}\right)^{1/4}\sqrt{|\mathcal{Y}||\mathcal{F}|\sqrt{2\log |\mathcal{F}|}}+\sqrt{2|\mathcal{Y}|\delta_{\mathcal{F}}}
\end{align*}
where $n_{F,C}$ is the number of periods $t$ where $C_t=C$ and $\pi_t\in F$.

Consider any two periods $t,\tau$ where $I_t=I_\tau$ but $C_t\neq C_\tau$. Since $I_t=I_\tau$ and information includes the forecast as context, we know that $\pi_t=\pi_\tau$. Now, consider
\[
n_{F_t,C_t}d_1(\pi_t,\hat{\pi}_{C_t})+n_{F_\tau,C_\tau}d_1(\pi_\tau,\hat{\pi}_{C_\tau})
\]
\[
=n_{F_t,C_t}d_1(\pi_t,\hat{\pi}_{C_t})+n_{F_t,C_\tau}d_1(\pi_t,\hat{\pi}_{C_\tau})
\]
\[
\geq \left(n_{F_t,C_t}+n_{F_t,C_\tau}\right)d_1\left(\pi_t,\frac{1}{n_{F_t,C_t}+n_{F_t,C_\tau}}\left(n_{F_t,C_t}\hat{\pi}_{C_t}+n_{F_t,C_\tau}d_1(\pi_t,\hat{\pi}_{C_\tau})\right)\right)
\]
by subadditivity and homogeneity of norms. By continuing this process of combining contexts, we find
\[
\frac{1}{T}\sum_{C\in\mathcal{C}}\sum_{F\in\mathcal{F}}n_{F,C}d_1(\pi_t,\hat{\pi}_C)
\geq\frac{1}{T}\sum_{I\in\mathcal{I}}n_Id_1(\pi_t,\hat{\pi}_I)
=\iota
\]
Therefore, the earlier miscalibration bound applies to $\ex[\sigma^*]{\iota}$ as well. Finally, we obtain our bound on the expected principal's regret.
\begin{align*}
\ex[\sigma^*]{\pregs}
\leq&\frac{1}{T}\sum_{I\in\mathcal{I}}n_I
\Delta(\hat{\pi}_I,\bar\epsilon)
+2\left(\frac{\epsilon+2\left(\frac{|\mathcal{C}|}{T}\right)^{1/4}\sqrt{|\mathcal{Y}||\mathcal{F}|\sqrt{2\log |\mathcal{F}|}}+2\sqrt{2|\mathcal{Y}|\delta_{\mathcal{F}}}+K^U_{\mathcal{R}}\delta_{\mathcal{R}}+K^U_{\mathcal{P}}\delta_{\mathcal{P}}}{\bar{\epsilon}}\right)\\
&+\left(2\left(\frac{|\mathcal{C}|}{T}\right)^{1/4}\sqrt{|\mathcal{Y}||\mathcal{F}|\sqrt{2\log |\mathcal{F}|}}+2\sqrt{2|\mathcal{Y}|\delta_{\mathcal{F}}}+2K^V_{\mathcal{R}}\delta_{\mathcal{R}}+K^V_{\mathcal{P}}\delta_{\mathcal{P}}\right)
\end{align*}

	\subsection{Proof of Theorem \ref{T5}}

Define
\[
\hat{\gamma}_P(I,y)=\frac{n_I\hat{\pi}_I(y)}{n_F\hat{\pi}_F(y)}\cdot\textbf{1}(I\in\mathcal{I}_F)
\]
as the empirical information structure conditional on forecast context $F$. This definition follows from Bayes' rule.

Assume access to a forecast $\pi_t\in\Delta(\mathcal{Y})$ for every period $t$. We will define this later. In period $t$, the mechanism computes the policy $p^*(\pi_t,\bar{\epsilon})$ that maximizes the worst-case payoff in the $\bar{\epsilon}$-robust stage game, treating the forecast $\pi_t$ as the common prior. That is,
\[
p^*(\pi_t,\bar{\epsilon})\in\arg\max_{p\in\mathcal{P}}\alpha_p(\pi_t,\bar{\epsilon})
\]
The mechanism chooses $p_t$ as follows. Let $P$ be the unique policy context that includes the policy $p^*(\pi_t,\bar{\epsilon})\in P$. Let $p_t=p_P$, where $p_P\in\mathcal{P}_1$ is the representative element of $P$.

The next two lemmas imply an upper bound on the principal's regret in terms of the quantity
\[
\iota\geq\frac{1}{T}\sum_{t=1}^Td_1(\pi_t,\hat{\pi}_F)
\]
that measures the discrepancy between the forecast $\pi_t$ and the empirical distribution $\hat{\pi}_F$ conditioned on the forecast context $F$. Lemma \ref{L4} is a lower bound on the principal's payoff under $\sigma^*$. Lemma \ref{L5} is an upper bound on the his payoff under any constant $\sigma^p\in\Sigma_0$.

\begin{lemma}\label{L4}
	Suppose the principal runs the mechanism $\sigma^*$. Then
	\begin{align*}
		\frac{1}{T}\sum_{t=1}^T
		V(r_t,p_t,y_t)
		\geq&
		\frac{1}{T}\sum_{F\in\mathcal{F}}n_F\max_{p\in\mathcal{P}}\alpha_{p}(\hat\pi_F,\bar\epsilon)
		-\left(\frac{\epsilon+6\iota+K^U_{\mathcal{R}}\delta_{\mathcal{R}}+2K^U_{\mathcal{P}}\delta_{\mathcal{P}}}{\bar{\epsilon}}\right)\\
		&-\left(M_1(\epsilon+\tilde\epsilon+2\iota+2K^U_{\mathcal{P}}\delta_{\mathcal{P}})+M_2
		+3\iota+K^V_{\mathcal{R}}\delta_{\mathcal{R}}+K^V_{\mathcal{P}}\delta_{\mathcal{P}}\right)
	\end{align*}
\end{lemma}
\begin{proof}
	Let $r_I$ be a representative element in the response context $R$ associated with information $I$ under mechanism $\sigma^*$. By regularity,
	\begin{align*}
		\epsilon_I
		&\geq\max_{r\in\mathcal{R}}\sum_{t\in I}\left(U(r,p_I,y_t)-U(r_I,p_I,y_t)\right)-K^U_{\mathcal{R}}\delta_{\mathcal{R}}\\
		&=\max_{r\in\mathcal{R}}\ex[y\sim\hat{\pi}_I]{U(r,p_I,y)-U(r_I,p_I,y)}-K^U_{\mathcal{R}}\delta_{\mathcal{R}}
	\end{align*}
	It follows, by regularity and definition of $\alpha$, that
	\begin{align*}
		\frac{1}{n_I}\sum_{t\in I}V(r_t,p_I,y_t)
		&\geq\ex[y\sim\hat{\pi}_I]{V(r_I,p_I,y)}-K^V_{\mathcal{R}}\delta_{\mathcal{R}}\\
		&\geq\alpha_{p_I}(\hat{\pi}_I,\epsilon_I+K^U_{\mathcal{R}}\delta_{\mathcal{R}})-K^V_{\mathcal{R}}\delta_{\mathcal{R}}
	\end{align*}
	Summing over information $I\in\mathcal{I}_F$ and using lemma \ref{L1}, we obtain
	\begin{align*}
		\frac{1}{n_F}\sum_{I\in\mathcal{I}_F}\sum_{t\in I}
		V(r_t,p_F,y_t)
		&\geq\frac{1}{n_F}\sum_{I\in\mathcal{I}_F}\sum_{t\in I}
		\alpha_{p_F}(\hat{\pi}_I,\epsilon_I+K^U_{\mathcal{R}}\delta_{\mathcal{R}})-K^V_{\mathcal{R}}\delta_{\mathcal{R}}\\
		&\geq\frac{1}{n_F}\sum_{I\in\mathcal{I}_F}\sum_{t\in I}\alpha_{p^*(\pi_t,\bar{\epsilon})}(\hat{\pi}_I,\epsilon_I+K^U_{\mathcal{R}}\delta_{\mathcal{R}}+2K^U_{\mathcal{P}}\delta_{\mathcal{P}})-K^V_{\mathcal{R}}\delta_{\mathcal{R}}-K^V_{\mathcal{P}}\delta_{\mathcal{P}}\\
		&\geq\frac{1}{n_F}\sum_{I\in\mathcal{I}_F}\sum_{t\in I}\left(\alpha_{p^*(\pi_t,\bar{\epsilon})}(\hat{\pi}_I,\bar\epsilon)-\frac{\epsilon_I+K^U_{\mathcal{R}}\delta_{\mathcal{R}}+2K^U_{\mathcal{P}}\delta_{\mathcal{P}}}{\bar{\epsilon}}\right)-K^V_{\mathcal{R}}\delta_{\mathcal{R}}-K^V_{\mathcal{P}}\delta_{\mathcal{P}}\\
		&=\frac{1}{n_F}\sum_{I\in\mathcal{I}_F}\sum_{t\in I}
		\alpha_{p^*(\pi_t,\bar{\epsilon})}(\hat{\pi}_I,\bar\epsilon)
		-\left(\frac{\epsilon_F+K^U_{\mathcal{R}}\delta_{\mathcal{R}}+2K^U_{\mathcal{P}}\delta_{\mathcal{P}}}{\bar{\epsilon}}\right)
		-K^V_{\mathcal{R}}\delta_{\mathcal{R}}-K^V_{\mathcal{P}}\delta_{\mathcal{P}}
	\end{align*}
	So far, we have a lower bound for the principal's payoff that nearly matches the principal's worst-case payoff in the stage game if the agent had information structure $\hat{\gamma}_F$ in each forecast context. Furthermore, we know that this information structure cannot be particularly useful to the agent. Define
	\[
	-\tilde\epsilon_F=\max_{r\in\mathcal{R}}\frac{1}{n_F}\sum_{I\in\mathcal{I}_F}\sum_{t\in I}\left(U(r,p_F,y_t)-U(r_t,p_F,y_t)\right)
	\]
	as the (possibly negative) regret accumulated in forecast context $F$ relative to the best-in-hindsight response, rather than the best-in-hindsight function from information to responses. Let $\pi_F$ be the forecast associated with forecast context $F$. Note that
	\begin{align*}
		\epsilon_F+\tilde\epsilon_F
		&=\min_{\tilde{r}\in\mathcal{R}}\frac{1}{n_F}\sum_{I\in\mathcal{I}_F}\max_{r\in\mathcal{R}}\sum_{t\in I}\left(U(r,p_F,y_t)-U(\tilde{r},p_F,y_t)\right)\\
		&\geq\min_{\tilde{r}\in\mathcal{R}}\frac{1}{n_F}\sum_{I\in\mathcal{I}_F}\max_{r\in\mathcal{R}}\sum_{t\in I}
		\left(U(r,p^*(\pi_F,\bar{\epsilon}),y_t)-U(\tilde{r},p^*(\pi_F,\bar{\epsilon}),y_t)\right)-2K^U_{\mathcal{P}}\delta_{\mathcal{P}}\\
		&=\min_{r}\max_{r_J}\ex[y\sim\hat{\pi}_F]{\ex[J\sim\hat{\gamma}_F(\cdot,y)]{U(r_J,p^*(\pi_F,\bar{\epsilon}),y)-U(r,p^*(\pi_F,\bar{\epsilon}),y)}}
		-2K^U_{\mathcal{P}}\delta_{\mathcal{P}}\\
		&\geq\min_{r}\max_{r_J}\ex[y\sim\pi_F]{\ex[J\sim\hat{\gamma}_F(\cdot,y)]{U(r_J,p^*(\pi_F,\bar{\epsilon}),y)-U(r,p^*(\pi_F,\bar{\epsilon}),y)}}
		-2d_1(\pi_F,\hat{\pi}_F)-2K^U_{\mathcal{P}}\delta_{\mathcal{P}}
	\end{align*}
	It follows from assumption \ref{A10} that
	\[
	M_1(\epsilon_F+\tilde\epsilon_F+2d_1(\pi_F,\hat{\pi}_F)+2K^U_{\mathcal{P}}\delta_{\mathcal{P}})+M_2
	\geq
	\alpha_{p^*(\pi_F,\bar\epsilon)}(\pi_F,\bar\epsilon)-\alpha_{p^*(\pi_F,\bar\epsilon)}(\pi_F,\hat\gamma_F,\bar\epsilon)
	\]
	which can be rewritten as
	\begin{equation}\label{E3}
	\alpha_{p^*(\pi_F,\bar\epsilon)}(\pi_F,\hat\gamma_F,\bar\epsilon)
	\geq\alpha_{p^*(\pi_F,\bar\epsilon)}(\pi_F,\bar\epsilon)-\left(M_1(\epsilon_F+\tilde\epsilon_F+2d_1(\pi_F,\hat{\pi}_F)+2K^U_{\mathcal{P}}\delta_{\mathcal{P}})+M_2\right)
	\end{equation}
	Next, we relate our lower bound on the principal's payoff to the term $\alpha_{p^*(\pi_F,\bar\epsilon)}(\pi_F,\hat\gamma_F,\bar\epsilon)$. Note that
	\begin{align*}
		\frac{1}{n_F}\sum_{I\in\mathcal{I}_F}n_I\alpha_{p^*(\pi_F,\bar{\epsilon})}(\hat{\pi}_I,\bar\epsilon)
		&=\ex[y\sim\hat{\pi}_F]{\ex[J\sim\hat{\gamma}(\cdot,y)]{\alpha_{p^*(\pi_F,\bar{\epsilon})}(\hat\pi_{F\mid J},\bar\epsilon)}}\\
		&\geq\min_{e_J}\ex[y\sim\hat\pi_F]{\ex[J\sim\hat{\gamma}(\cdot,y)]{\alpha_{p^*(\pi_F,\bar{\epsilon})}(\pi_{F\mid J},e_J)}}
		\quad\mathrm{s.t.}\quad\bar{\epsilon}=\ex[y\sim\pi_F]{\ex[J\sim\hat\gamma{\cdot,y}]{e_J}}\\
		&=\alpha_{p^*(\pi_F,\bar\epsilon)}(\hat\pi_F,\hat\gamma_F,\bar\epsilon)\\
		&\geq\alpha_{p^*(\pi_F,\bar\epsilon)}(\pi_F,\hat\gamma_F,\bar\epsilon)-\frac{2d_1(\pi_F,\hat{\pi}_F)}{\bar{\epsilon}}-d_1(\pi_F,\hat{\pi}_F)
	\end{align*}
	So combining this with inequality \eqref{E3} gives
	\begin{align*}
		&\frac{1}{n_F}\sum_{I\in\mathcal{I}_F}n_I\alpha_{p^*(\pi_F,\bar{\epsilon})}(\hat{\pi}_I,\bar\epsilon)\\
		&\geq
		\alpha_{p^*(\pi_F,\bar\epsilon)}(\pi_F,\bar\epsilon)
		-\left(M_1(\epsilon_F+\tilde\epsilon_F+2d_1(\pi_F,\hat{\pi}_F)+2K^U_{\mathcal{P}}\delta_{\mathcal{P}})+M_2\right)
		-\left(\frac{2d_1(\pi_F,\hat{\pi}_F)}{\bar{\epsilon}}\right)
		-d_1(\pi_F,\hat{\pi}_F)\\
		&\geq
		\alpha_{p^*(\hat\pi_F,\bar\epsilon)}(\hat\pi_F,\bar\epsilon)
		-\left(M_1(\epsilon_F+\tilde\epsilon_F+2d_1(\pi_F,\hat{\pi}_F)+2K^U_{\mathcal{P}}\delta_{\mathcal{P}})+M_2\right)
		-\left(\frac{6d_1(\pi_F,\hat{\pi}_F)}{\bar{\epsilon}}\right)
		-3d_1(\pi_F,\hat{\pi}_F)
	\end{align*}
	Collapsing these inequalities gives us
	\begin{align*}
		\frac{1}{n_F}\sum_{t\in F}
		V(r_t,p_F,y_t)
		\geq&
		\alpha_{p^*(\hat\pi_F,\bar\epsilon)}(\hat\pi_F,\bar\epsilon)
		-\left(\frac{\epsilon_F+6d_1(\pi_F,\hat{\pi}_F)+K^U_{\mathcal{R}}\delta_{\mathcal{R}}+2K^U_{\mathcal{P}}\delta_{\mathcal{P}}}{\bar{\epsilon}}\right)\\
		&-\left(M_1(\epsilon_F+\tilde\epsilon_F+2d_1(\pi_F,\hat{\pi}_F)+2K^U_{\mathcal{P}}\delta_{\mathcal{P}})+M_2
		+3d_1(\pi_F,\hat{\pi}_F)+K^V_{\mathcal{R}}\delta_{\mathcal{R}}+K^V_{\mathcal{P}}\delta_{\mathcal{P}}\right)
	\end{align*}
	Summing over forecast contexts $F\in\mathcal{F}$ gives us the desired result.
\end{proof}

\begin{lemma}\label{L5}
	Suppose the principal runs some constant mechanism $\sigma^p\in\Sigma_0$. Then
	\[
	\frac{1}{T}\sum_{t=1}^TV(r_t,p,y_t)
	\leq\frac{1}{T}\sum_{F\in\mathcal{F}}n_F\max_{\tilde{p}\in\mathcal{P}}\beta_{\tilde{p}}(\hat\pi_F,\bar\epsilon)
	+\left(\frac{\epsilon+K^U_{\mathcal{R}}\delta_{\mathcal{R}}}{\bar{\epsilon}}\right)
	+\left(M_1(\epsilon+\tilde\epsilon)+M_2+K^V_{\mathcal{R}}\delta_{\mathcal{R}}\right)
	\]
\end{lemma}
\begin{proof}
	Let $r_I$ be a representative element in the response context $R$ associated with information $I$ under mechanism $\sigma^p$. By regularity,
	\begin{align*}
		\epsilon_I
		&\geq\max_{r\in\mathcal{R}}\sum_{t\in I}\left(U(r,p,y_t)-U(r_I,p,y_t)\right)-K^U_{\mathcal{R}}\delta_{\mathcal{R}}\\
		&=\max_{r\in\mathcal{R}}\ex[y\sim\hat{\pi}_I]{U(r,p,y)-U(r_I,p,y)}-K^U_{\mathcal{R}}\delta_{\mathcal{R}}
	\end{align*}
	It follows, by regularity and definition of $\beta$, that
	\begin{align*}
		\frac{1}{n_I}\sum_{t\in I}V(r_t,p,y_t)
		&\leq\ex[y\sim\hat{\pi}_I]{V(r_I,p,y)}+K^V_{\mathcal{R}}\delta_{\mathcal{R}}\\
		&\leq\beta_{p}(\hat{\pi}_I,\epsilon_I+K^U_{\mathcal{R}}\delta_{\mathcal{R}})+K^V_{\mathcal{R}}\delta_{\mathcal{R}}
	\end{align*}
	Summing over information $I\in\mathcal{I}_F$ and using lemma \ref{L1}, we obtain
	\begin{align*}
		\frac{1}{n_F}\sum_{I\in\mathcal{I}_F}\sum_{t\in I}
		V(r_t,p,y_t)
		&\leq\frac{1}{n_F}\sum_{I\in\mathcal{I}_F}n_I
		\beta_{p}(\hat{\pi}_I,\epsilon_I+K^U_{\mathcal{R}}\delta_{\mathcal{R}})+K^V_{\mathcal{R}}\delta_{\mathcal{R}}\\
		&\leq\frac{1}{n_F}\sum_{I\in\mathcal{I}_F}n_I
		\left(\beta_{p}(\hat{\pi}_I,\bar\epsilon)+\frac{\epsilon_I+K^U_{\mathcal{R}}\delta_{\mathcal{R}}}{\bar{\epsilon}}\right)
		+K^V_{\mathcal{R}}\delta_{\mathcal{R}}\\
		&=\frac{1}{n_F}\sum_{I\in\mathcal{I}_F}n_I
		\beta_{p}(\hat{\pi}_I,\bar\epsilon)
		+\left(\frac{\epsilon_F+K^U_{\mathcal{R}}\delta_{\mathcal{R}}}{\bar{\epsilon}}\right)
		+K^V_{\mathcal{R}}\delta_{\mathcal{R}}
	\end{align*}
	So far, we have an upper bound for the principal's payoff that nearly matches the principal's worst-case payoff in the stage game if the agent had information structure $\hat{\gamma}_F$ in each forecast context. Furthermore, we know that this information structure cannot be particularly useful to the agent. Define
	\[
	-\tilde\epsilon_F=\max_{r\in\mathcal{R}}\frac{1}{n_F}\sum_{I\in\mathcal{I}_F}\sum_{t\in I}\left(U(r,p,y_t)-U(r_t,p,y_t)\right)
	\]
	as the (possibly negative) regret accumulated in forecast context $F$ relative to the best-in-hindsight response, rather than the best-in-hindsight function from information to responses. Let $\pi_F$ be the forecast associated with forecast context $F$. Note that
	\begin{align*}
		\epsilon_F+\tilde\epsilon_F
		&=\min_{\tilde{r}\in\mathcal{R}}\frac{1}{n_F}\sum_{I\in\mathcal{I}_F}\max_{r\in\mathcal{R}}\sum_{t\in I}\left(U(r,p,y_t)-U(\tilde{r},p,y_t)\right)
	\end{align*}
	It follows from assumption \ref{A10} that
	\[
	M_1(\epsilon_F+\tilde\epsilon_F)+M_2
	\geq\beta_{p}(\hat\pi_F,\hat\gamma_F,\bar\epsilon)-\max_{\tilde{p}\in\mathcal{P}}\beta_{\tilde{p}}(\hat\pi_F,\bar\epsilon)
	\]
	which can be rewritten as
	\[
	\beta_{p}(\hat\pi_F,\hat\gamma_F,\bar\epsilon)
	\leq\max_{\tilde{p}\in\mathcal{P}}\beta_{\tilde{p}}+M_1(\epsilon_F+\tilde\epsilon_F)+M_2
	\]
	Next, we relate our upper bound on the principal's payoff to the term $\beta_{p}(\hat\pi_F,\hat\gamma_F,\bar\epsilon)$. Note that
	\begin{align*}
		\frac{1}{n_F}\sum_{I\in\mathcal{I}_F}n_I\beta_{p}(\hat{\pi}_I,\bar\epsilon)
		&=\ex[y\sim\hat{\pi}_F]{\ex[J\sim\hat{\gamma}(\cdot,y)]{\beta_{p}(\hat\pi_{F\mid J},\bar\epsilon)}}\\
		&\leq\max_{e_J}\ex[y\sim\hat\pi_F]{\ex[J\sim\hat{\gamma}(\cdot,y)]{\beta_{p}(\pi_{F\mid J},e_J)}}\quad\mathrm{s.t.}\quad\bar{\epsilon}=\ex[y\sim\pi_F]{\ex[J\sim\hat{\cdot,y}]{e_J}}\\
		&=\beta_{p}(\hat\pi_F,\hat\gamma_F,\bar\epsilon)
	\end{align*}
	Collapsing these inequalities gives us
	\[
	\frac{1}{n_F}\sum_{t\in F}
	V(r_t,p,y_t)
	\leq\max_{\tilde{p}\in\mathcal{P}}\beta_{\tilde{p}}(\hat\pi_F,\bar\epsilon)
	+\left(\frac{\epsilon_F+K^U_{\mathcal{R}}\delta_{\mathcal{R}}}{\bar{\epsilon}}\right)
	+\left(M_1(\epsilon_F+\tilde\epsilon_F)+M_2+K^V_{\mathcal{R}}\delta_{\mathcal{R}}\right)
	\]
	Summing over forecast contexts $F\in\mathcal{F}$ gives us the desired result.
\end{proof}

From these two lemmas, it immediately follows that
\begin{align*}
	\pregs
	\leq&\frac{1}{T}\sum_{F\in\mathcal{F}}n_F
	\Delta(\hat{\pi}_F,\bar\epsilon)
	+\left(\frac{2\epsilon+6\iota+2K^U_{\mathcal{R}}\delta_{\mathcal{R}}+2K^U_{\mathcal{P}}\delta_{\mathcal{P}}}{\bar{\epsilon}}\right)\\
	&+\left(M_1(2\epsilon+2\tilde\epsilon+2\iota+2K^U_{\mathcal{P}}\delta_{\mathcal{P}})+2M_2+3\iota+2K^V_{\mathcal{R}}\delta_{\mathcal{R}}+K^V_{\mathcal{P}}\delta_{\mathcal{P}}\right)
\end{align*}
Therefore, to bound the principal's regret, all that remains is to bound $\iota$. If we use the algorithm from appendix \ref{Ap2-0} to generate $\pi_t$, this follows directly from equation \eqref{E4}, which states
\[
\ex[\sigma^*]{\iota}
\leq
\sqrt{|\mathcal{Y}||\mathcal{F}|\sqrt{\frac{2\log |\mathcal{F}|}{T}}+2|\mathcal{Y}|\delta_{\mathcal{F}}}
\]
Finally, we obtain our bound on the expected principal's regret.
\begin{align*}
	\ex[\sigma^*]{\pregs}
	\leq&\frac{1}{T}\sum_{F\in\mathcal{F}}n_F
	\Delta(\hat{\pi}_F,\bar\epsilon)
	+\left(\frac{2\epsilon+6\sqrt{|\mathcal{Y}||\mathcal{F}|\sqrt{\frac{2\log |\mathcal{F}|}{T}}+2|\mathcal{Y}|\delta_{\mathcal{F}}}
	+2K^U_{\mathcal{R}}\delta_{\mathcal{R}}+2K^U_{\mathcal{P}}\delta_{\mathcal{P}}}{\bar{\epsilon}}\right)\\
	&+M_1\left(2\epsilon+2\tilde\epsilon
	+2\sqrt{|\mathcal{Y}||\mathcal{F}|\sqrt{\frac{2\log |\mathcal{F}|}{T}}+2|\mathcal{Y}|\delta_{\mathcal{F}}}
	+2K^U_{\mathcal{P}}\delta_{\mathcal{P}}\right)\\
	&+2M_2
	+3\sqrt{|\mathcal{Y}||\mathcal{F}|\sqrt{\frac{2\log |\mathcal{F}|}{T}}+2|\mathcal{Y}|\delta_{\mathcal{F}}}
	+2K^V_{\mathcal{R}}\delta_{\mathcal{R}}+K^V_{\mathcal{P}}\delta_{\mathcal{P}}
\end{align*}

	\subsection{Proof of Theorem \ref{T6}}

Assume access to a forecast $\pi_t$ for every period $t$. We will define this later. The mechanism chooses $p_t$ as follows. Let $P$ be the unique policy context that includes the policy $p^\dagger(\pi_t,\bar{\epsilon})\in P$. Let $p_t=p_P$, where $p_P\in\mathcal{P}_1$ is the representative element of $P$.

The next two lemmas imply an upper bound on the principal's regret in terms of the quantity
\[
\iota\geq\frac{1}{T}\sum_{t=1}^Td_1(\pi_t,\hat{\pi}_F)
\]
that measures the discrepancy between the forecast $\pi_t$ and the empirical distribution $\hat{\pi}_F$ conditioned on the forecast context $F$. Lemma \ref{L6} is a lower bound on the principal's payoff under $\sigma^*$. Lemma \ref{L7} is an upper bound on the his payoff under any constant $\sigma^p\in\Sigma_0$.

\begin{lemma}\label{L6}
	Suppose the principal runs the mechanism $\sigma^*$. Then
	\begin{align*}
		\frac{1}{T}\sum_{t=1}^TV(r_t,p_t,y_t)
		\geq&
		\frac{1}{T}\sum_{F\in\mathcal{F}}n_F\inf_{\gamma}\max_{p\in\mathcal{P}}\alpha_{p}(\hat\pi_F,\gamma,\bar\epsilon)
		-\left(\frac{\epsilon+4\iota+K^U_{\mathcal{R}}\delta_{\mathcal{R}}+2K^U_{\mathcal{P}}\delta_{\mathcal{P}}}{\bar{\epsilon}}\right)\\
		&-\left(2\iota+K^V_{\mathcal{R}}\delta_{\mathcal{R}}+K^V_{\mathcal{P}}\delta_{\mathcal{P}}\right)
	\end{align*}
\end{lemma}
\begin{proof}
	Let $r_I$ be a representative element in the response context $R$ associated with information $I$ under mechanism $\sigma^*$. By regularity,
	\begin{align*}
		\epsilon_I
		&\geq\max_{r\in\mathcal{R}}\sum_{t\in I}\left(U(r,p_I,y_t)-U(r_I,p_I,y_t)\right)-K^U_{\mathcal{R}}\delta_{\mathcal{R}}\\
		&=\max_{r\in\mathcal{R}}\ex[y\sim\hat{\pi}_I]{U(r,p_I,y)-U(r_I,p_I,y)}-K^U_{\mathcal{R}}\delta_{\mathcal{R}}
	\end{align*}
	It follows, by regularity and definition of $\alpha$, that
	\begin{align*}
		\frac{1}{n_I}\sum_{t\in I}V(r_t,p_I,y_t)
		&\geq\ex[y\sim\hat{\pi}_I]{V(r_I,p_I,y)}-K^V_{\mathcal{R}}\delta_{\mathcal{R}}\\
		&\geq\alpha_{p_I}(\hat{\pi}_I,\epsilon_I+K^U_{\mathcal{R}}\delta_{\mathcal{R}})-K^V_{\mathcal{R}}\delta_{\mathcal{R}}
	\end{align*}
	Summing over information $I\in\mathcal{I}_F$, we obtain
	\begin{align*}
		\frac{1}{n_F}\sum_{I\in\mathcal{I}_F}\sum_{t\in I}
		V(r_t,p_F,y_t)
		&\geq\frac{1}{n_F}\sum_{I\in\mathcal{I}_F}\sum_{t\in I}
		\alpha_{p_F}(\hat{\pi}_I,\epsilon_I+K^U_{\mathcal{R}}\delta_{\mathcal{R}})-K^V_{\mathcal{R}}\delta_{\mathcal{R}}\\
		&\geq\frac{1}{n_F}\sum_{I\in\mathcal{I}_F}\sum_{t\in I}\alpha_{p^\dagger(\pi_t,\bar\epsilon)}(\hat{\pi}_I,\epsilon_I+K^U_{\mathcal{R}}\delta_{\mathcal{R}}+2K^U_{\mathcal{P}}\delta_{\mathcal{P}})-K^V_{\mathcal{R}}\delta_{\mathcal{R}}-K^V_{\mathcal{P}}\delta_{\mathcal{P}}\\
		&\geq\frac{1}{n_F}\sum_{I\in\mathcal{I}_F}\sum_{t\in I}\left(\alpha_{p^\dagger(\pi_t,\bar\epsilon)}(\hat{\pi}_I,\bar\epsilon)-\frac{\epsilon_I+K^U_{\mathcal{R}}\delta_{\mathcal{R}}+2K^U_{\mathcal{P}}\delta_{\mathcal{P}}}{\bar{\epsilon}}\right)-K^V_{\mathcal{R}}\delta_{\mathcal{R}}-K^V_{\mathcal{P}}\delta_{\mathcal{P}}\\
		&=\frac{1}{n_F}\sum_{I\in\mathcal{I}_F}\sum_{t\in I}
		\alpha_{p^\dagger(\pi_t,\bar\epsilon)}(\hat{\pi}_I,\bar\epsilon)
		-\left(\frac{\epsilon_F+K^U_{\mathcal{R}}\delta_{\mathcal{R}}+2K^U_{\mathcal{P}}\delta_{\mathcal{P}}}{\bar{\epsilon}}\right)
		-K^V_{\mathcal{R}}\delta_{\mathcal{R}}-K^V_{\mathcal{P}}\delta_{\mathcal{P}}
	\end{align*}
	Focus on the first term, i.e.
	\begin{align*}
		\frac{1}{n_F}\sum_{I\in\mathcal{I}_F}n_I\alpha_{p^\dagger(\pi_F,\bar\epsilon)}(\hat{\pi}_I,\bar\epsilon)
		&=\ex[y\sim\hat{\pi}_F]{\ex[J\sim\hat{\gamma}(\cdot,y)]{\alpha_{p^\dagger(\pi_F,\bar\epsilon)}(\hat\pi_{F\mid J},\bar\epsilon)}}\\
		&\geq\min_{e_J}\ex[y\sim\hat\pi_F]{\ex[J\sim\hat{\gamma}(\cdot,y)]{\alpha_{p^\dagger(\pi_F,\bar\epsilon)}(\hat\pi_{F\mid J},e_J)}}
		\quad\mathrm{s.t.}\quad\bar{\epsilon}=\ex[y\sim\hat\pi_F]{\ex[J\sim\hat\gamma{\cdot,y}]{e_J}}\\
		&=\alpha_{p^\dagger(\pi_F,\bar\epsilon)}(\hat\pi_F,\hat\gamma_F,\bar\epsilon)\\
		&\geq\inf_\gamma\alpha_{p^\dagger(\pi_F,\bar\epsilon)}(\hat\pi_F,\gamma,\bar\epsilon)\\
		&\geq\inf_\gamma\alpha_{p^\dagger(\hat\pi_F,\bar\epsilon)}(\hat\pi_F,\gamma,\bar\epsilon)-\frac{4d_1(\pi_F,\hat{\pi}_F)}{\bar\epsilon}-2d_1(\pi_F,\hat{\pi}_F)
	\end{align*}
	Collapsing these inequalities gives us
	\begin{align*}
	\frac{1}{n_F}\sum_{I\in\mathcal{I}_F}\sum_{t\in I}V(r_t,p_F,y_t)
	\geq&
	\inf_\gamma\alpha_{p^\dagger(\hat\pi_F,\bar\epsilon)}(\hat\pi_F,\gamma,\bar\epsilon)
	-\left(\frac{\epsilon_F+4d_1(\pi_F,\hat{\pi}_F)+K^U_{\mathcal{R}}\delta_{\mathcal{R}}+2K^U_{\mathcal{P}}\delta_{\mathcal{P}}}{\bar{\epsilon}}\right)\\
	&-\left(2d_1(\pi_F,\hat{\pi}_F)+K^V_{\mathcal{R}}\delta_{\mathcal{R}}+K^V_{\mathcal{P}}\delta_{\mathcal{P}}\right)
	\end{align*}
	Summing over forecast contexts $F\in\mathcal{F}$ gives us the desired result.
\end{proof}

\begin{lemma}\label{L7}
	Suppose the principal runs some constant mechanism $\sigma^p\in\Sigma_0$. Then
	\[
	\frac{1}{T}\sum_{t=1}^TV(r_t,p,y_t)
	\leq
	\sum_{F\in\mathcal{F}}n_F\max_{\tilde{p}\in\mathcal{P}}\beta_{\tilde{p}}(\hat\pi_F,\hat\gamma_F,\bar\epsilon)
	+\left(\frac{\epsilon+K^U_{\mathcal{R}}\delta_{\mathcal{R}}}{\bar{\epsilon}}\right)
	+K^V_{\mathcal{R}}\delta_{\mathcal{R}}
	\]
\end{lemma}
\begin{proof}
	Let $r_I$ be a representative element in the response context $R$ associated with information $I$ under mechanism $\sigma^p$. By regularity,
	\begin{align*}
		\epsilon_I
		&\geq\max_{r\in\mathcal{R}}\sum_{t\in I}\left(U(r,p,y_t)-U(r_I,p,y_t)\right)-K^U_{\mathcal{R}}\delta_{\mathcal{R}}\\
		&=\max_{r\in\mathcal{R}}\ex[y\sim\hat{\pi}_I]{U(r,p,y)-U(r_I,p,y)}-K^U_{\mathcal{R}}\delta_{\mathcal{R}}
	\end{align*}
	It follows, by regularity and definition of $\beta$, that
	\begin{align*}
		\frac{1}{n_I}\sum_{t\in I}V(r_t,p,y_t)
		&\leq\ex[y\sim\hat{\pi}_I]{V(r_I,p,y)}+K^V_{\mathcal{R}}\delta_{\mathcal{R}}\\
		&\leq\beta_{p}(\hat{\pi}_I,\epsilon_I+K^U_{\mathcal{R}}\delta_{\mathcal{R}})+K^V_{\mathcal{R}}\delta_{\mathcal{R}}
	\end{align*}
	Summing over information $I\in\mathcal{I}_F$, we obtain
	\begin{align*}
		\frac{1}{n_F}\sum_{I\in\mathcal{I}_F}\sum_{t\in I}
		V(r_t,p,y_t)
		&\leq\frac{1}{n_F}\sum_{I\in\mathcal{I}_F}n_I
		\beta_{p}(\hat{\pi}_I,\epsilon_I+K^U_{\mathcal{R}}\delta_{\mathcal{R}})+K^V_{\mathcal{R}}\delta_{\mathcal{R}}\\
		&\leq\frac{1}{n_F}\sum_{I\in\mathcal{I}_F}n_I
		\left(\beta_{p}(\hat{\pi}_I,\bar\epsilon)+\frac{\epsilon_I+K^U_{\mathcal{R}}\delta_{\mathcal{R}}}{\bar{\epsilon}}\right)
		+K^V_{\mathcal{R}}\delta_{\mathcal{R}}\\
		&=\frac{1}{n_F}\sum_{I\in\mathcal{I}_F}n_I
		\beta_{p}(\hat{\pi}_I,\bar\epsilon)
		+\left(\frac{\epsilon_F+K^U_{\mathcal{R}}\delta_{\mathcal{R}}}{\bar{\epsilon}}\right)
		+K^V_{\mathcal{R}}\delta_{\mathcal{R}}
	\end{align*}
	Focus on the first term, i.e.
	\begin{align*}
		\frac{1}{n_F}\sum_{I\in\mathcal{I}_F}n_I\beta_{p}(\hat{\pi}_I,\bar\epsilon)
		&=\ex[y\sim\hat{\pi}_F]{\ex[J\sim\hat{\gamma}(\cdot,y)]{\beta_{p}(\hat\pi_{F\mid J},\bar\epsilon)}}\\
		&\leq\max_{e_J}\ex[y\sim\hat\pi_F]{\ex[J\sim\hat{\gamma}(\cdot,y)]{\beta_{p}(\hat\pi_{F\mid J},e_J)}}
		\quad\mathrm{s.t.}\quad\bar{\epsilon}=\ex[y\sim\hat\pi_F]{\ex[J\sim\hat\gamma{\cdot,y}]{e_J}}\\
		&=\beta_{p}(\hat\pi_F,\hat\gamma_F,\bar\epsilon)\\
		&\leq\max_{\tilde{p}\in\mathcal{P}}\beta_{\tilde{p}}(\hat\pi_F,\hat\gamma_F,\bar\epsilon)
	\end{align*}
	Collapsing these inequalities gives us
	\begin{align*}
		\frac{1}{n_F}\sum_{I\in\mathcal{I}_F}\sum_{t\in I}V(r_t,p,y_t)
		\leq
		\max_{\tilde{p}\in\mathcal{P}}\beta_{\tilde{p}}(\hat\pi_F,\hat\gamma_F,\bar\epsilon)
		+\left(\frac{\epsilon_F+K^U_{\mathcal{R}}\delta_{\mathcal{R}}}{\bar{\epsilon}}\right)
		+K^V_{\mathcal{R}}\delta_{\mathcal{R}}
	\end{align*}
	Summing over forecast contexts $F\in\mathcal{F}$ gives us the desired result.
\end{proof}

From these two lemmas, it immediately follows that
\begin{align*}
	\pregs
	\leq
	\sum_{F\in\mathcal{F}}n_F\nabla(\hat\pi_F,\bar\epsilon)
	+2\left(\frac{\epsilon+2\iota+K^U_{\mathcal{R}}\delta_{\mathcal{R}}+K^U_{\mathcal{P}}\delta_{\mathcal{P}}}{\bar{\epsilon}}\right)
	+\left(2\iota+2K^V_{\mathcal{R}}\delta_{\mathcal{R}}+K^V_{\mathcal{P}}\delta_{\mathcal{P}}\right)
\end{align*}
Therefore, to bound the principal's regret, all that remains is to bound $\iota$. If we use the algorithm from appendix \ref{Ap2-0} to generate $\pi_t$, this follows directly from equation \eqref{E4}, which states
\[
\ex[\sigma^*]{\iota}
\leq
\sqrt{|\mathcal{Y}||\mathcal{F}|\sqrt{\frac{2\log |\mathcal{F}|}{T}}+2|\mathcal{Y}|\delta_{\mathcal{F}}}
\]
Finally, we obtain our bound on the expected principal's regret.
\begin{align*}
	\ex[\sigma^*]{\pregs}
	\leq&\frac{1}{T}\sum_{F\in\mathcal{F}}n_F
	\nabla(\hat{\pi}_F,\bar\epsilon)
	+2\left(\frac{\epsilon+2\sqrt{|\mathcal{Y}||\mathcal{F}|\sqrt{\frac{2\log |\mathcal{F}|}{T}}+2|\mathcal{Y}|\delta_{\mathcal{F}}}
		+K^U_{\mathcal{R}}\delta_{\mathcal{R}}+K^U_{\mathcal{P}}\delta_{\mathcal{P}}}{\bar{\epsilon}}\right)\\
	&+\left(2\sqrt{|\mathcal{Y}||\mathcal{F}|\sqrt{\frac{2\log |\mathcal{F}|}{T}}+2|\mathcal{Y}|\delta_{\mathcal{F}}}+2K^V_{\mathcal{R}}\delta_{\mathcal{R}}+K^V_{\mathcal{P}}\delta_{\mathcal{P}}\right)
\end{align*}

	\subsection{Proof of Theorem \ref{T1}}
	
	We adapt the proof of theorem \ref{T4} to prove theorem \ref{T1}. This will require only relatively minor changes. Let $\hat{\pi}_{I,F}$ denote the empirical distribution among periods $t\in I\cap F$. Let $n_{I,F}$ indicate the number of such periods. Let $\pi_F$ denote the (unique) forecast associated with forecast context $F$.  Previously, we defined
	\[
	\iota\geq\frac{1}{T}\sum_{t=1}^Td_1(\pi_t,\hat{\pi}_I)
	\]
	Now, we define
	\[
	\iota\geq\frac{1}{T}\sum_{t=1}^Td_1(\pi_t,\hat{\pi}_{I,F})
	\]
	Begin at the last line of lemma \ref{L2}, where it says ``focus on the first term''. Rewrite that first term as
	\[
	\frac{1}{T}\sum_{I\in\mathcal{I}}\sum_{F\in\mathcal{F}}\sum_{t\in I\cap F}
	\alpha_{p^*(\pi_F,\bar{\epsilon})}(\hat{\pi}_I,\bar\epsilon)
	\]
	Now we switch $\hat{\pi}_I$ with $\pi_I$, i.e. the convex combination of forecasts,
	\[
	\pi_I=\frac{1}{n_I}\sum_{F\in\mathcal{F}}n_{I,F}\pi_F
	\]
	By lemma \ref{???1}, this gives us
	\[
	\frac{1}{T}\sum_{I\in\mathcal{I}}\sum_{F\in\mathcal{F}}\sum_{t\in I\cap F}
	\alpha_{p^*(\pi_F,\bar{\epsilon})}(\hat{\pi}_I,\bar\epsilon)
	\]
	\[
	\geq\frac{1}{T}\sum_{I\in\mathcal{I}}\sum_{F\in\mathcal{F}}\sum_{t\in I\cap F}
	\left(\alpha_{p^*(\pi_F,\bar{\epsilon})}(\pi_I,\bar\epsilon)-\frac{2d_1(\pi_I,\hat{\pi}_I)}{\bar{\epsilon}}-d_1(\pi_I,\hat{\pi}_I)\right)
	\]
	Note that every forecast $\pi_F$ leads to a policy $p^*(\pi_F,\bar{\epsilon})$ that is in the policy context $P$ associated with information $I$. By assumption \ref{???},
	\[
	\geq\frac{1}{T}\sum_{I\in\mathcal{I}}\sum_{F\in\mathcal{F}}\sum_{t\in I\cap F}
	\left(\alpha_{p^*(\pi_I,\bar{\epsilon})}(\pi_I,\bar\epsilon)-\frac{2d_1(\pi_I,\hat{\pi}_I)}{\bar{\epsilon}}-d_1(\pi_I,\hat{\pi}_I)-O(\delta_{\mathcal{P}})\right)
	\]
	Now we apply lemma \ref{???1} again,
	\[
	\geq\frac{1}{T}\sum_{I\in\mathcal{I}}\sum_{F\in\mathcal{F}}\sum_{t\in I\cap F}
	\left(\alpha_{p^*(\pi_I,\bar{\epsilon})}(\hat\pi_I,\bar\epsilon)-\frac{4d_1(\pi_I,\hat{\pi}_I)}{\bar{\epsilon}}-2d_1(\pi_I,\hat{\pi}_I)-O(\delta_{\mathcal{P}})\right)
	\]
	and then lemma \ref{???2}
	\[
	\geq\frac{1}{T}\sum_{I\in\mathcal{I}}\sum_{F\in\mathcal{F}}\sum_{t\in I\cap F}
	\left(\alpha_{p^*(\hat\pi_I,\bar{\epsilon})}(\hat\pi_I,\bar\epsilon)-\frac{8d_1(\pi_I,\hat{\pi}_I)}{\bar{\epsilon}}-4d_1(\pi_I,\hat{\pi}_I)-O(\delta_{\mathcal{P}})\right)
	\]
	By the homogeneity and subadditivity of the $l_1$ norm,
	\[
	d_1(\pi_I,\hat{\pi}_I)\leq\frac{1}{n_I}\sum_{i=1}^nn_{I,F}d_1(\pi_{I,F},\hat{\pi}_{I,F})
	\]
	which gives us
	\[
	\geq\frac{1}{T}\sum_{I\in\mathcal{I}}\sum_{F\in\mathcal{F}}\sum_{t\in I\cap F}
	\left(\alpha_{p^*(\hat\pi_I,\bar{\epsilon})}(\hat\pi_I,\bar\epsilon)-\frac{8d_1(\pi_{I,F},\hat{\pi}_{I,F})}{\bar{\epsilon}}-4d_1(\pi_{I,F},\hat{\pi}_{I,F})-O(\delta_{\mathcal{P}})\right)
	\]
	\[
	\geq\frac{1}{T}\sum_{I\in\mathcal{I}}\sum_{F\in\mathcal{F}}\sum_{t\in I\cap F}
	\left(\alpha_{p^*(\hat\pi_I,\bar{\epsilon})}(\hat\pi_I,\bar\epsilon)\right)
	-\frac{8\iota}{\bar{\epsilon}}-4\iota-O(\delta_{\mathcal{P}})
	\]
	This is essentially where we were by the end of lemma \ref{L2}, with the addition of an $O(\delta_{\mathcal{P}})$ term and slightly different constants.
	
	Lemma \ref{L3} requires no change. The discussion following lemma \ref{L3} requires very little change. Find the line that begins with ``Consider any two periods''. We rewrite as follows. Consider any two periods $t,\tau$ where $I_t=I_\tau$ and $F_t=F_\tau$ but $C_t\neq C_\tau$. Since $F_t=F_\tau$ we know that $\pi_t=\pi_\tau$. Now, consider
	\[
	n_{F_t,C_t}d_1(\pi_t,\hat{\pi}_{C_t})+n_{F_\tau,C_\tau}d_1(\pi_\tau,\hat{\pi}_{C_\tau})
	\]
	\[
	=n_{F_t,C_t}d_1(\pi_t,\hat{\pi}_{C_t})+n_{F_t,C_\tau}d_1(\pi_t,\hat{\pi}_{C_\tau})
	\]
	\[
	\geq \left(n_{F_t,C_t}+n_{F_t,C_\tau}\right)d_1\left(\pi_t,\frac{1}{n_{F_t,C_t}+n_{F_t,C_\tau}}\left(n_{F_t,C_t}\hat{\pi}_{C_t}+n_{F_t,C_\tau}d_1(\pi_t,\hat{\pi}_{C_\tau})\right)\right)
	\]
	by subadditivity and homogeneity of norms. By continuing this process of combining contexts, we find
	\[
	\frac{1}{T}\sum_{C\in\mathcal{C}}\sum_{F\in\mathcal{F}}n_{F,C}d_1(\pi_t,\hat{\pi}_C)
	\geq\frac{1}{T}\sum_{I\in\mathcal{I}}\sum_{F\in\mathcal{F}}n_{I,F}d_1(\pi_t,\hat{\pi}_{I,F})
	=\iota
	\]
	Therefore, our bound holds except with the addition of an $O(\delta_{\mathcal{P}})$ term and slightly different constants.
	
	\subsection{Proof of Theorem \ref{T2}}
	
	Define 
	\[
	\pi_P=\frac{1}{n_P}\sum_{F\in\mathcal{F}}n_{P,F}\pi_F
	\]
	It is straightforward to adapt the proof of theorem \ref{T5}. Replace all reference to $p^*(\pi_F,\bar{\epsilon})$ with $p^*(\pi_P,\bar{\epsilon})$. This changes $U$ and $V$ (and all derived terms, like $\alpha$) by at most $O(\delta_{\mathcal{P}})$, by assumption \ref{???regularity3}. Replace all remaining references of forecast contexts $F$ to policy contexts $P$. It remains to verify that
	\[
	\iota\geq\frac{1}{T}\sum_{F\in\mathcal{F}}\sum_{t\in F}d_1(\pi_F,\hat{\pi}_F)\geq\frac{1}{T}\sum_{P\in\mathcal{P}}\sum_{t\in F}d_1(\pi_P,\hat{\pi}_P)
	\]
	which follows from the homogeneity and subadditivity of the $l_1$ norm, and the fact that $\pi_P,\hat{\pi}_P$ are convex combinations of $\pi_F,\hat{\pi}_F$ for $F\subseteq P$.
	
	\subsection{Proof of Theorem \ref{T3}}
	
	Define 
	\[
	\pi_P=\frac{1}{n_P}\sum_{F\in\mathcal{F}}n_{P,F}\pi_F
	\]
	It is straightforward to adapt the proof of theorem \ref{T6}. Replace all reference to $p^\dagger(\pi_F,\bar{\epsilon})$ with $p^\dagger(\pi_P,\bar{\epsilon})$. This changes $U$ and $V$ (and all derived terms, like $\alpha$) by at most $O(\delta_{\mathcal{P}})$, by assumption \ref{???regularity3}. Replace all remaining references of forecast contexts $F$ to policy contexts $P$. It remains to verify that
	\[
	\iota\geq\frac{1}{T}\sum_{F\in\mathcal{F}}\sum_{t\in F}d_1(\pi_F,\hat{\pi}_F)\geq\frac{1}{T}\sum_{P\in\mathcal{P}}\sum_{t\in F}d_1(\pi_P,\hat{\pi}_P)
	\]
	which follows from the homogeneity and subadditivity of the $l_1$ norm, and the fact that $\pi_P,\hat{\pi}_P$ are convex combinations of $\pi_F,\hat{\pi}_F$ for $F\subseteq P$.

\end{document}